\documentclass[12pt]{article}

\usepackage{amsfonts,amsmath,amsthm,amsbsy,amssymb,dsfont,bm,mathtools,mathalfa,mathrsfs}
\usepackage[hidelinks]{hyperref} 
\usepackage{graphicx}
\usepackage{enumerate}
\usepackage[shortlabels,inline]{enumitem} 
\usepackage{natbib}
\usepackage{multicol} 
\usepackage[boxed,noline,noend]{algorithm2e} 
\SetKwInput{KwData}{Input}
\SetKwInput{KwResult}{Output} 
\usepackage{setspace}
\usepackage{authblk}

\newcommand{\1}{\textbf{\textup{1}}}
\newcommand{\I}{\mathds{1}} 
\newcommand{\R}{\mathbb{R}}

\newcommand{\N}{\mathbb{N}}

\newcommand{\km}{{k_\mathrm{max}}}
\newcommand{\T}{^\mathrm{T}}
\newcommand{\test}{_\mathrm{test}}

\newcommand{\norm}[1]{{\left\|#1\right\|}}
\newcommand{\F}{_\mathrm{F}}
\newcommand{\mumax}{\mu_{\mathrm{max}}}
\newcommand{\balpha}{\boldsymbol{\alpha}}

\DeclareMathOperator*{\argmax}{argmax}
\DeclareMathOperator*{\argmin}{argmin}
\DeclareMathOperator{\tr}{tr}
\DeclareMathOperator{\Var}{Var}
\DeclareMathOperator{\E}{\mathbb{E}}
\DeclareMathOperator{\Prob}{\mathbb{P}}
\DeclareMathOperator{\diag}{diag}

\newtheorem{assumption}{Assumption}
\newtheorem{corollary}{Corollary}[section]
\newtheorem{definition}{Definition}

\newtheorem{lemma}{Lemma}[section]
\newtheorem{prop}{Proposition}[section]
\newtheorem{remark}{Remark}[section]
\newtheorem{theorem}{Theorem}[section]

\newcommand{\beginsupplement}{%
	\setcounter{section}{0} 
	\setcounter{table}{0}
	\setcounter{figure}{0} 
	\setcounter{algocf}{0}
	\renewcommand{\thesection}{S\arabic{section}}
	\renewcommand{\thetable}{S\arabic{table}}
	\renewcommand{\thefigure}{S\arabic{figure}}
	\renewcommand{\thealgocf}{S\arabic{algocf}}
	\renewcommand{\theequation}{S\arabic{equation}}
	\renewcommand{\thetheorem}{S\arabic{theorem}}
}

\begin{document}

\def\spacingset#1{\renewcommand{\baselinestretch}%
    {#1}\small\normalsize} \spacingset{1}


\title{\bf Estimating Graph Dimension with Cross-validated Eigenvalues}
\author[1]{Fan Chen}
\author[2]{S\'{e}bastien Roch} 
\author[1]{Karl Rohe}
\author[2]{Shuqi Yu}
\affil[1]{Department of Statistics, University of Wisconsin--Madison}
\affil[2]{Department of Mathematics, University of Wisconsin--Madison}
\maketitle

\bigskip
\begin{abstract}
    In applied multivariate statistics, estimating the number of latent dimensions or the number of clusters, $k$, is a fundamental and recurring problem. We study a sequence of statistics called \textit{cross-validated eigenvalues}. Under a large class of random graph models, including both Poisson and Bernoulli edges, without parametric assumptions, we provide a $p$-value for each cross-validated eigenvalue. It tests the null hypothesis that the sample eigenvector is orthogonal to (i.e., uncorrelated with) the true latent dimensions. This approach naturally adapts to problems where some dimensions are not statistically detectable. In scenarios where all $k$ dimensions can be estimated, we show that our procedure consistently estimates $k$. In simulations and data example, the proposed estimator compares favorably to alternative approaches in both computational and statistical performance.
\end{abstract}

\noindent%
{\it Keywords:} central limit theorem, cross-validation, graph dimension, random grap
\vfill

\newpage
\spacingset{1.15} 

\section{Introduction}

In social network analysis, a large and popular class of models supposes that each person has a set of $k$ latent characteristics and the probability that a pair of people are friends depends only on that pair's $k$ characteristics.  Typically, for example, if two people have  similar characteristics, then they are more likely to become friends. We aim to estimate the number of characteristics $k$ using a class of models where every edge is statistically independent, conditionally on the characteristics.  This includes the Latent Space Model, the Aldous-Hoover representation, and graphons \citep{hoff2002latent, aldous1985exchangeability, hoover1989tail, lovasz2012large, jacobs2014unified}.

We are particularly interested in the class of models called the random dot product model \citep{rdpg}, which includes the Stochastic Blockmodel, along with its degree-corrected and mixed membership variants \citep{karrer_stochastic_2011, airoldi2008mixed}.  
In the Stochastic Blockmodel, $k$ is the number of blocks.

\subsection{Motivating examples} \label{sec:motivation_figure}

Denote the adjacency matrix $A \in \N^{n \times n}$ as recording the number of edges between $i$ and $j$ in element $A_{ij}$.  
In the models that we consider $A$ is a random matrix and $\E(A)$ 
has $k$ non-zero eigenvalues. We use this fact to estimate $k$.

The eigenvalues of $A$ (i.e., the ``sample eigenvalues'') in the scree plot unfortunately often fail to provide a clear estimate of $k$. 
As an illustration, Figure \ref{fig:motivation_scree} shows such cases in simulated and real data. 
In the left panel, the simulated graph comes from a Degree-Corrected Stochastic Blockmodel, with $k=128$ hierarchically arranged blocks. 
Many of the 128 dimensions cannot be estimated from the data. As such, it is not surprising that no artifacts arise in the scree plot around 128.  
The right panel gives the scree plot for a citation graph of 22,688 academic journals. This graph was constructed from the Semantic Scholar database \citep{ammar2018construction}, where each journal is a node, and citations between two papers represent edges. In this graph, the average node degree is 35.  
Supplementary Section \ref{sec:motivation_details} provides more details on these two graphs. In Figure \ref{fig:motivation_scree}, neither panel provides clear guidance for how to choose $k$. This is often the fact with scree plots.  Figures \ref{fig:motivation_simulation} and \ref{fig:motivation_data} provide an alternative plot that gives clearer guidance.

\begin{figure}[] 
	\centering
	\textbf{
    We often look for an elbow or a gap in scree plots.}
	\vspace{.06in}
	\includegraphics[width=5.5in]{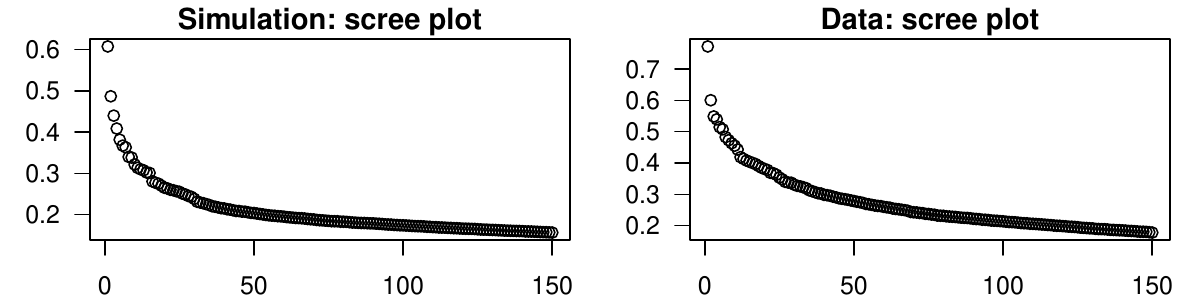} 
	\caption{In these examples, it is difficult to detect a gap or elbow. In the left panel, the graph is simulated from a Degree-Corrected Stochastic Blockmodel with $n=2560$.  In the right panel, the graph is a citation graph among $n = 22,688$ academic journals.  Displayed are the largest 150 eigenvalues of the normalized and regularized adjacency matrix. 
    }
	\label{fig:motivation_scree}
\end{figure}

In the left panel of Figure \ref{fig:motivation_simulation}, the black line is the scree plot from the left panel of Figure \ref{fig:motivation_scree}. In that same panel, the orange line gives the eigenvalues of $\E(A)$, i.e. the ``population eigenvalues''; notice how it goes to zero at $128$. In this panel, the black line is significantly greater than the orange line.  This is why the standard scree plot is so difficult to use. That gap comes from \textit{the bias of the sample eigenvalues}. This paper explains how this bias arises from ``overfitting'' to the noise in the random graph.  

\begin{figure}[] 
	\centering
	\includegraphics[width=5.5in]{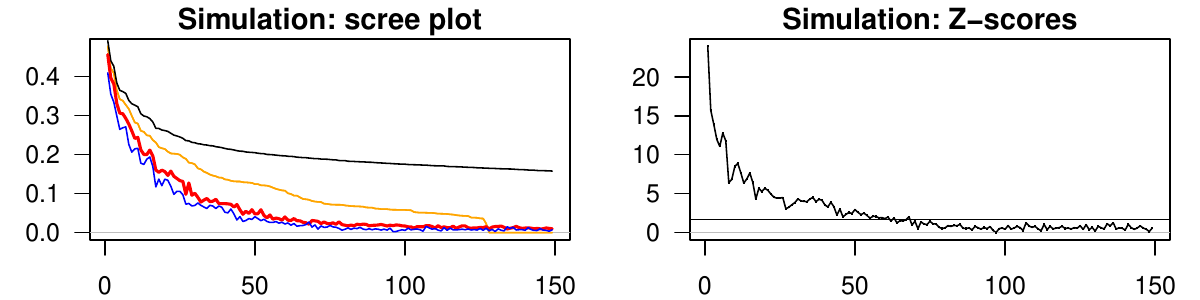} 
	\caption{In the left panel, the black line gives the sample eigenvalues (repeated from the left panel of Figure \ref{fig:motivation_scree}) and the orange line gives the $k=128$ non-zero population eigenvalues. The first two eigenvalues have been removed to improve the display. The blue line gives cross-validated eigenvalues, and the red line gives their population version, cross-population eigenvalue. In the right panel, the black line depicts the Z-scores for cross-validated eigenvalues.  The horizontal black line corresponds to the cutoff for 0.05 significance level (two-side). In this example, a good choice for $\hat k$ would be around 60.}
	\label{fig:motivation_simulation}
\end{figure}

The blue line in the left panel of Figure \ref{fig:motivation_simulation} gives our proposed \textit{cross-validated eigenvalues}.
The red line gives the population version of such quantities, named \textit{cross-population eigenvalues}, where the $\E(A)$ is assumed known in place of $A$. 
The red and blue lines reveal that the signal strength of eigenvectors computed from the data diminishes after around $\hat k \approx 60$. Interestingly, for each cross-validated eigenvalue, one can test the null hypothesis its population counter part (the cross-population eigenvalue) is equal to zero.  So, the right panel of Figure \ref{fig:motivation_simulation} gives the Z-score for each of the cross-validated eigenvalues.  
In that panel, the Z-scores start to cluster around the cutoff at $\hat k \approx 60$. 
The standard scree plot (black line in left panel) does not reveal anything around this value. 

\begin{figure}[] 
	\centering
	\includegraphics[width=5.5in]{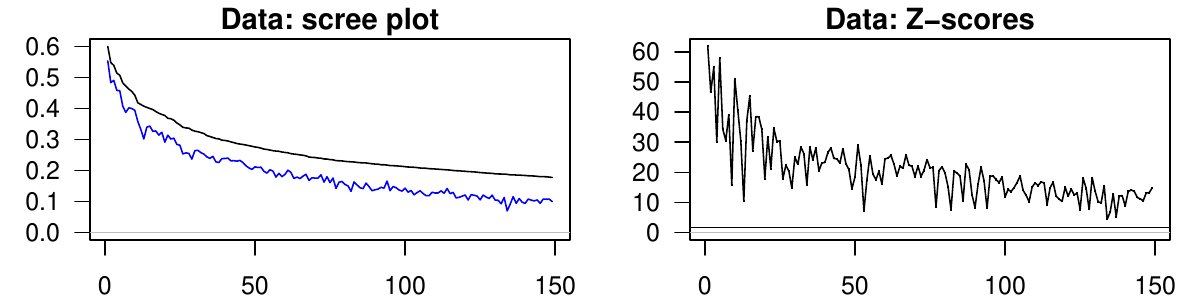} 
	\caption{In the left panel, the black line gives the empirical eigenvalues (repeated from the right panel of Figure \ref{fig:motivation_scree}). The blue line gives cross-validated eigenvalues. In the right panel, the Z-scores are calculated under the null hypothesis that each cross-population eigenvalue is zero.  The horizontal black line corresponds to the cutoff at 0.05 significance level.}
	\label{fig:motivation_data}
\end{figure}

In Figure \ref{fig:motivation_data}, we return to the citation graph of 22,688 academic journals displayed in the right panel of Figure \ref{fig:motivation_scree}.  
In this example, the cross-validated eigenvalues suggest that all of the leading 150 dimensions are highly statistically significant. This is consistent with the results in \cite{rohe2020vintage} that showed the leading 100 dimensions reveal groups of journals that form coherent academic areas.  
It's worth pointing out that $A$ has 789,980 non-zero elements, spread across 22,688 rows and columns. Despite the relatively large size of this graph, computing all 150 cross-validated eigenvalues and their Z-scores requires less than 10 seconds in \texttt{R} on a 2.0GHz quad‑core processor (Intel Core i5, MacBook Pro 2020).  This performance is enabled by the sparse matrix packages \texttt{Matrix} and \texttt{RSpectra} \citep{rMatrix, RSpectra}.

\subsection{Our contributions}

In this paper, we exploit a notion of cross-validated eigenvalues as a new approach to estimating $k$. 
A key piece for this technique is to split the edges in the graph into two independent sets \citep{abbe_community_2015, abbe_exact_2016}. 
Then, the eigenvectors and cross-validated eigenvalues are defined using the two graphs.  This removes the bias from overfitting. 
Under a large class of random graph models, we provide a simple procedure to compute cross-validated eigenvalues.

A related holdout approach was previously explored in the econometrics literature \citep{abadir_design_2014,lam_nonparametric_2016} for covariance estimation. 
In this paper and in those prior papers, the eigenvectors are estimated with a portion of the data and the ``signal strength'' of those vectors is estimated with the remaining held-out data. There are three key differences with this previous work.  First, the observed data is of a different nature; for a covariance matrix $\Sigma \in \R^{ p \times p}$, it is assumed in \cite{lam_nonparametric_2016} that we observe $y_i = \Sigma^{1/2} x_i$, where $x_i$ are unobserved and contain independent, identically distributed (i.i.d.)~random variables. Second, the notion of sample splitting is different; the approach in \cite{abadir_design_2014} constructs a sample covariance matrix with a subsample of the observed vectors $y_1, \dots, y_n$.  Third, the target of estimation is different; we provide $p$-values to estimate $k$, while the prior work aims to estimate the eigenvalues of $\Sigma$, which is presumed to be full rank.    

Since the preprint version of this paper \citep{chen2021estimating}, related sample-splitting approaches have been developed for settings where traditional sample splitting is infeasible. \cite{neufeld2024data,neufeld2024inference} propose ``data thinning'' and ``count splitting,'' methods that decompose individual observations into independent parts for convolution-closed distributions, enabling model validation tasks including dimension selection for low-rank approximations. \cite{baharav2024oasis} develop finite-sample valid tests for contingency tables using linear test statistics. Most recently, \cite{ancell2025post} apply data thinning to network data, using one network realization to estimate community structure and another to conduct inference---a setting that parallels our use of sample splitting for estimating graph dimension.

For cross-validated eigenvalues, we provide a central limit theorem, which leads to a $p$-value for the statistical significance of a sample eigenvector. This can be used to estimate the number of statistically useful sample eigenvectors, and thus $k$. 
We provide consistency results for the proposed estimator of $k$, allowing for weighted and sparse graphs.
Finally, through simulations and real data applications, we show that this estimator compares favorably to alternative approaches in both computational and statistical performance.

\subsection{Prior literature}

Numerous methods have been proposed to estimate $k$ under the Stochastic Blockmodel and its degree-corrected version \citep{bordenave_non-backtracking_2015,bickel_hypothesis_2016,lei_goodness--fit_2016,wang_likelihood-based_2017,chen_network_2018,ma_determining_2019,le_estimating_2019,liu_community_2019,jin_estimating_2020}. One previous technique has been proposed to estimate the dimension of the more general random dot product graph \citep{li_network_2020}.
These methods roughly fall into one of three categories: spectral, cross-validation, and (penalized) likelihood-based approaches.
Methods based on likelihood or cross-validation are actively researched, yet the majority of them are commonly restrained by the scale of networks. 
Spectral methods are highly scalable for estimating $k$ in large networks, although their rigorous analyses require delicate, highly technical random matrix arguments \citep{ajanki_universality_2017,benaych-georges_largest_2019,chakrabarty_eigenvalues_2020,dumitriu_sparse_2019,benaych-georges_spectral_2020,hwang_local_2020}.

Among the likelihood-based approaches, the authors of \cite{wang_likelihood-based_2017} proposed to estimate $k$ by solving a Bayesian information criterion (BIC) type optimization problem, where the objective function is a sum of the log-likelihood and of the model complexity. The computation is often not feasible because the likelihood contains exponentially many terms. In \cite{ma_determining_2019}, a pseudo-likelihood ratio is used to compare the goodness-of-fit of models with differing $k$s that have been estimated using spectral clustering with regularization \citep{rohe_spectral_2011,qin2013regularized,joseph_impact_2016,su_strong_2019}, speeding up the computation. However, the two methods allow little node degree heterogeneity. Related to the goodness-of-fit technique, the authors of \cite{jin_estimating_2020} present a stepwise testing approach based on the number of quadrilaterals in the networks. Computing the statistic requires at least $n^2$ multiplication operations, regardless of the sparsity of the graph, thus is infeasible for large $n$.  
More recently, cross-validation \citep{picard_cross-validation_1984,arlot_survey_2010} has also been adapted to the context of choosing $k$. 
For example, in \cite{chen_network_2018}, a block-wise node-pair splitting technique is introduced. In each fold, a block of rows of the adjacency matrix are held out from the Stochastic Blockmodel fitting (including the community memberships), then the left-out rows are used to calculate a predictive loss. In \cite{li_network_2020}, the authors propose to hold out a random fraction of node-pairs, instead of nodes (thus all the incidental node-pairs). 
In addition, they suggest using a general low-rank matrix completion (e.g., a singular value thresholding approach \citep{chatterjee_matrix_2015}) to calculate the loss on the left-out node-pairs. Theoretical conditions for not under-estimating $k$ were established in both cross-validation based methods \citep{chen_network_2018,li_network_2020}. 
Calculating the loss on either held-out rows or on scattered values in the adjacency matrix requires $O(n^2)$ computations, regardless of sparsity. This limits the ability of these techniques to scale to large graphs. 

In \cite{bickel_hypothesis_2016,lei_goodness--fit_2016}, hypothesis tests using the top eigenvalue or singular value of a properly normalized adjacency matrix are proposed, based on edge universality and other related results for general Wigner ensembles \citep{tracy_level-spacing_1994,soshnikov_universality_1999,erdos_spectral_2012,erdos_local_2013,alex_isotropic_2014}. The analyses of these hypothesis tests assume dense graphs. In \cite{liu_community_2019}, a version of the ``elbow in the scree plot'' approach (see, e.g., \cite{zhu_automatic_2006} for a discussion of this approach) is analyzed rigorously under the Degree-Corrected Stochastic Blockmodel, also in the dense case. For sparser graphs, the spectral properties of other matrices associated to graphs have been used to estimate $k$, including the non-backtracking matrix \citep{krzakala_spectral_2013,bordenave_non-backtracking_2015,le_estimating_2019} and the Bethe-Hessian matrix \citep{le_estimating_2019}. However, their theoretical analysis currently allow little node degree heterogeneity in the sparse case.

There is also related work on bootstrapping \citep{snijders1999non,thompson_using_2016,green_bootstrapping_2017,levin_bootstrapping_2019,lin_higher-order_2020}, jackknife resampling \citep{lin_theoretical_2020} and subsampling \citep{bhattacharyya_subsampling_2015,lunde_subsampling_2019,naulet_bootstrap_2021} in network analysis. In particular, in \cite{lunde_subsampling_2019}, subsampling schemes are applied to the nonzero eigenvalues of the adjacency matrix under low-rank graphon models. Weak convergence results are established under some technical conditions, including sufficient edge density (i.e., average degree growing asymptotically faster than $\sqrt{n}$); simulation results also indicate that sparsity leads to poor performance for the estimators considered, especially in the case of the eigenvalues closer to the bulk. 

\section{Measuring the signal strength of sample eigenvectors}

We consider a connected multigraph $G = (V,E)$ consisting of the set of nodes $V = \{1, \dots, n\}$ and undirected edges 
$E$, where we allow multiple edges and self-loops.
The \textit{adjacency matrix} $A \in \N^{n \times n}$ records the number of edges between $i$ and $j$ in element $A_{ij}$.   
Define the \textit{population (or expectation)} matrix of $A$ as $P = \E(A)$. 
In this paper, we focus on a simple class of random dot product graphs (see remarks after the definitions).

\begin{definition}[Poisson random graph] \label{def:poisson} 
	Let $A$ be the adjacency matrix of a Poisson random graph. Then, $A$ is symmetric, and the upper-triangular elements of $A$ are independent Poisson random variables, and its population matrix $P$ has the eigendecomposition
	\begin{equation} \label{eq:Amodel} 
		P = U \Lambda U\T
	\end{equation}
	for $U \in \R^{n \times k}$ with orthonormal columns and diagonal matrix $\Lambda \in \R^{k \times k}$ with positive elements $\lambda_1, \dots, \lambda_k \in \R$ down the diagonal in non-increasing order.
\end{definition}

While Definition \ref{def:poisson} makes two simplifying assumptions, the method generalizes to graph with various edges, which we remark on in order.

\begin{remark}[Directed edges]
The first simplifying assumption is that $A$ is symmetric (i.e., edges are undirected).  This assumption can be relaxed; Remark \ref{remark:rectangle} discusses directed graphs, contingency tables, and rectangular incidence matrices. 
\end{remark}

\begin{remark}[Bernoulli edges]
The second simplifying assumption is that the elements of $A$ are Poisson. 
It is common to model $A_{ij}$ as Poisson because it is convenient and, in sparse graphs specifically, the difference between the Poisson and Bernoulli models becomes negligible \citep{karrer_stochastic_2011, flynn2020profile,crane2018edge, cai2016edge, amini_biclustering}.  
Section \ref{sec:clt-bernoulli} provides the formal extension to Bernoulli graphs via a conditional CLT, with simulations in Supplementary Section \ref{sec:poisson_vs_bernoulli} demonstrating that the proposed technique provides reliable $p$-values.
\end{remark}

\begin{remark}[Random dot product model]
The introduction motivated the paper by expressing  $\E(A_{ij})= \langle Z_i, Z_j\rangle = \sum_{\ell=1}^k Z_{i\ell} Z_{j \ell}$, where $Z_i \in \R^k$ is a vector of characteristics for node $i$. 
If the $Z_1, \dots, Z_n \in \R^k$ span $\R^k$ and the elements $A_{ij}$ are independent Poisson variables, then this model is included in Definition \ref{def:poisson}.
\end{remark}

The diagonal of $\Lambda$ contains the leading $k$ \textit{population eigenvalues} of $P$ and their corresponding eigenvectors are in the columns of $U$.  
For $j>k$, the population eigenvalues of $P$ are $\lambda_j = 0$. 
In this paper, we aim to estimate the number of nonzero eigenvalues.

\subsection{Sample eigenvalues: a poor diagnostic}

A common approach to estimating the eigenvalues of $P$ is to use a plug-in estimator, i.e., estimating the  eigenvalues of $P$ with the eigenvalues of $A$.
Specifically, for $j= 1, \dots, n$, the $j$th \textit{eigenvectors} 
$\hat x_j \in \R^n$ of a symmetric matrix $A \in \N^{n \times n}$ is defined as  
\begin{eqnarray} \label{eq:eig}
	\hat x_j &=& \argmax_{x\in \hat S_j} \ \ x\T A x, 
\end{eqnarray}
where $\hat S_j = \{x \in R^n: \|x\|_2 = 1 \mbox{ and } x\T\hat x_\ell = 0 \mbox{ for } \ell = 1, \dots, j-1\}$.
Correspondingly, the objective values for $j= 1, \dots, n$ are called the \textit{eigenvalues} $\hat{\lambda}_j$,
\begin{equation}\label{eq:sample_eigenvalue}
	\hat{\lambda}_j =  \hat x_j\T A \hat x_j.
\end{equation}
The  eigenvalues of $A$, i.e., $\hat{\lambda}_1, \hat{\lambda}_2, \hat{\lambda}_3, \dots$, are often plotted against their index $1, 2, 3, \dots$.  This is called a scree plot and it is used as a diagnostic to estimate $k$ (Figure \ref{fig:motivation_scree}).  In this scree plot, there might be a ``gap" or an ``elbow" at the $k$th eigenvalue, which reveals $k$.

However, there is an overfitting and bias problem with the plug-in estimator for the population eigenvalues which can make the ``gap'' or ``elbow'' in the scree plot more difficult to observe. 
The leading eigenvalue estimates $\hat{\lambda}_1, \dots, \hat{\lambda}_k$ are asymptotically unbiased, so long as their corresponding population eigenvalues $\lambda_1, \dots, \lambda_k$ are large enough (see, e.g., \cite{chakrabarty_eigenvalues_2020} for related results).  
By contrast, $\lambda_{k+1}=0$ is not growing; $\hat{\lambda}_{k+1}$ is a biased estimate of $\lambda_{k+1}0$, with $\E(\hat{\lambda}_{k+1}) > \lambda_{k+1}=0$ (see, e.g., \cite{benaych-georges_spectral_2020} for related results).  
So, if $\lambda_k$ is not large enough, then this bias diminishes the appearance of a ``gap'' or ``elbow'' between $\hat \lambda_k$ and $\hat \lambda_{k+1}$ in the scree plot.

This can be seen in the left panel of Figure \ref{fig:motivation_simulation}.  In that figure, the gap between the sample eigenvalues (black line) and the population eigenvalues (orange line) grows wider as $k$ increases.\footnote{In this figure, it is not the eigenvalues of $A$ and $\E(A)$ that are shown, but rather the eigenvalues of the normalized and regularized forms of $A$ and $\E(A)$.}

\subsection{Population-validated eigenvalues}

Even if an oracle were to tell us that the population eigenvector $x_j$ has population eigenvalue $\lambda_j\ne 0,$ we should only use the sample eigenvector $\hat x_j$ for statistical inference if $\hat x_j$ is close to $x_j$ and other leading population eigenvectors.  Because of this requirement, the population eigenvalues $\lambda_j$ do not measure the signal strength of $\hat x_j$.

In the realistic setting where the signal is not overwhelming for every sample eigenvector $\hat x_j$ for $j<k$, we propose to measure the signal strength of this vector in the following way.
We define the \textit{population-validated eigenvalue}\footnote{The notion of $\lambda_P(\hat x_j)$ was previously proposed and studied in \cite{abadir_design_2014} and \cite{lam_nonparametric_2016} for optimal estimation of eigenvalue shrinkage under a different statistical model.}
of a sample eigenvector $\hat x_j$ as
\begin{equation}\label{eq:lambdaP}
	\lambda_{P}(\hat x_j)= \hat x_j\T P \hat x_j.
\end{equation}
Although this is not an eigenvalue in the traditional sense, there are three reasons to think of this quantity as an analogue.

First, the quadratic form in Equation \eqref{eq:lambdaP} mimics the form for the sample eigenvalue in Equation \eqref{eq:sample_eigenvalue}, which is also the objective function that the sample eigenvectors optimize in Equation \eqref{eq:eig}.  As such, for a population eigenvector $x_j$, $\lambda_{P}(x_j)$ is the corresponding population eigenvalue $\lambda_j$.

The second reason to think of $\lambda_P$ as an analogue of an eigenvalue is that for a sample eigenvector $\hat x$, 
there always exists a vector $\hat x^\perp \in \R^n$ that is orthogonal to $\hat x$ and
$P \hat x = \hat \lambda \hat x +  \hat x^\perp$
for some value $\hat \lambda \in \R.$ 
If $\hat x^\perp =0$, then $\hat x$ is an eigenvector of $P$ with eigenvalue $\hat \lambda$.  Even when $\hat x^\perp \ne 0$,
if $\|\hat x\|_2 = 1$, then $\hat \lambda = \lambda_{P}(\hat x )$. 
This is because
\[\lambda_{P}(\hat x )= \hat x \T P \hat x  =  \hat x \T (\hat \lambda \hat x +  \hat x^\perp) = \hat \lambda\]

The third reason also provides intuition for why $\lambda_P$ is a measure of signal strength; indeed 
$\lambda_P$ provides for the optimal reconstruction of $P$ as conceptualized in the following proposition (a proof is included in Supplementary Section \ref{sec:proof} for completeness). 
\begin{prop}[\cite{lam_nonparametric_2016}] \label{prop:signal_strength}
For $j = 1, \dots, q$, $\hat \lambda_j = \lambda_{P}(\hat x_j)$ is the solution to
$$\min_{\hat \lambda_1, \dots, \hat \lambda_q} \left\| P - \sum_{j = 1}^q \hat \lambda_j \hat x_j \hat x_j\T\right\|_F.$$
\end{prop}

A key advantage of the population quantity $\lambda_P(\hat x_j)$ over $\lambda_j$ is that it reveals information about the vector that we can compute, i.e., $\hat x_j$, not the eigenvector that we wish we had, i.e., $x_j$.  If a sample eigenvector $\hat x_j$ is close to its population counterpart $x_j$, then for the reasons above, it is reasonable to presume that $\lambda_{P}(\hat x_j)$ is close to the population eigenvalue $\lambda_j$. However, if the estimation problem is too difficult and $\hat x_j$ is nearly orthogonal to the eigenvectors of $P$ that have non-zero eigenvalues, then $\lambda_{P}(\hat x_j) \approx 0$. 
Notably, and most importantly, this can happen even if $\hat x_j$'s corresponding population eigenvector $x_j$ has a non-zero eigenvalue.  For example, this happens when the estimation problem is too difficult as happens for the 60th-128th sample eigenvectors in Figure \ref{fig:motivation_simulation}.


\section{The three key pieces for cross-validated eigenvalues}
\label{sec:splitting}

It is tempting to compute the eigenvectors $\hat x_j$ of $A$ with all of the edges and then choose $k$ based on the population-validated eigenvalues $\lambda_P(\hat x_j)$. 
However, $\lambda_P(\hat x_j)$ is not observed (as $P$ is unobserved). 
In this section, we detail the three key pieces that enable an unbiased estimator of population-validated eigenvalues, which we called \textit{cross-validation eigenvalues}. 
For cross-validation, we describe an edge splitting procedure to create two Poisson random graphs that are independent and have identical population eigenvectors. The cross-validation eigenvalue is an intuitive plug-in estimator using information from both split graphs, and mostly importantly, it is asymptotically normal.

\subsection{The first piece: edge splitting creates independent Poisson graphs}

The edge splitting (\texttt{ES}) procedure is described in Algorithm \ref{alg:edge_splitting}. \texttt{ES} splits the edges of a graph into two graphs and outputs the two adjacency matrices $\tilde A$ and $\tilde A\test$ (we omit ``train'' in $\tilde A_\text{train}$ for simplicity). 
Notice that under \texttt{ES} and conditionally on $A_{ij}$, $\tilde A_{ij} \sim \text{Binomial}(A_{ij}, 1-\varepsilon)$ and $[\tilde A\test]_{ij} = A_{ij} - \tilde A_{ij}.$ 

\begin{algorithm}
    \setstretch{1.35}
	\DontPrintSemicolon
	\caption{Edge splitting}
	\label{alg:edge_splitting}
	\KwData{adjacency matrix $A \in \mathbb{N}^{n \times n}$ and edge splitting probability $\varepsilon \in (0,1)$.\;}
	\textbf{Procedure} \texttt{ES}$(A,\varepsilon)$\textbf{:}\;
	\Indp
	1. Convert $A$ into $G = (V,E)$, where $\{i,j\}$ is repeated in the edge set $E$ potentially more than once if $A_{ij}>1$.\;
	2. Initiate $\tilde E\test$ and $\tilde E$, two empty edge sets on $V$.\;
	3. \For{each copy of edge $\{i,j\} \in E$}{assign it to $\tilde E\test$ with probability $\varepsilon$. Otherwise, assign it to $\tilde E$.}
	4. Convert $(V,\tilde E\test)$ into an adjacency matrix $\tilde A\test \in \mathbb{N}^{n \times n}$ and $(V,\tilde E)$ into an adjacency matrix $\tilde A \in \mathbb{N}^{n \times n}$\;
	\Indm
	\KwResult{$\tilde A$ and $\tilde A\test$.}
\end{algorithm}

For the first piece, the independence of $\tilde A$ and $\tilde A\test$ follows from the next lemma, often referred to as thinning (see, e.g., \citet[Section 3.7.2]{durrett_2019}).

\begin{lemma}\label{lemma:poisson_binomial}
	Define $X \sim \text{Poisson}(\lambda)$ and conditionally on $X$, define $Y \sim \text{Binomial}(X, p)$ and $Z = X-Y$.  Unconditionally on $X$, the random variables $Y$ and $Z$ are independent Poisson random variables and, further, $Y \sim \text{Poisson}(p \lambda)$ and $Z \sim \text{Poisson}((1-p) \lambda)$. 
\end{lemma}

To apply the lemma, let $X$ be $A_{ij}$, let $\lambda$ be $P_{ij}$, let $Y$ and $Z$ be the $(i,j)$-th elements of $\tilde A$ and $\tilde A\test$ respectively, and let $p = \varepsilon$. This results in the independence between $\tilde A$ and $\tilde A\test$.

\begin{corollary}
    $\tilde A$ and $\tilde A\test$ are the adjacency matrices of two independent Poisson random graphs, and $\E(\tilde A)=(1-\varepsilon) A$, and $\E(\tilde A\test)=\varepsilon A$.
\end{corollary}

\subsection{The second piece: the splitting population graphs have identical eigenvectors}

The next proposition shows that \texttt{ES} preserves the spectral properties of the population adjacency matrices $\E(\tilde A)$ and $\E(\tilde A\test)$ from $\E(A)$. This result does not require any distributional assumptions on $A$, only that its elements are integers (so that $\tilde A$ and $\tilde A\test$ are defined).

\begin{prop} \label{prop:split}   
	If $\tilde A$ and $\tilde A\test \in \mathbb{N}^{n \times n}$ are generated by applying \texttt{ES} to $A \in \mathbb{N}^{n \times n}$ with splitting probability $\varepsilon$, then
	\begin{enumerate}
		\item The eigenvectors of $\E(A), \E(\tilde A),$ and $\E(\tilde A\test)$ are  identical.
		\item If $\lambda_j$ is an eigenvalue of $\E(A)$, then $(1-\varepsilon) \lambda_j$ is an eigenvalue of $\E(\tilde A)$ and $\varepsilon \lambda_j$ is an eigenvalue of $\E(\tilde A\test)$.
	\end{enumerate}
	Here, all expectations are unconditional on $A$.
\end{prop}

\begin{proof}
	Define $P = \E(A)$ and let $P=U \Lambda U\T$ be its eigendecomposition; if $A$ is not random, then $P=A$ and $U, \Lambda$ potentially have $n$ columns. It follows directly from the construction in \texttt{ES} that $\E(\tilde A) = (1-\varepsilon) P$ and $\E(\tilde A\test) = \varepsilon P$. 
	Rearranging terms reveals the eigendecomposition of $\E(\tilde A\test)$,
	\[\E(\tilde A\test) = \varepsilon P =   U (\varepsilon \Lambda) U\T\]
	and similarly for $\E(\tilde A)$.
	This shows that they have the same eigenvectors and the simple relationship between their eigenvalues in the statement.  
\end{proof}

By making $\varepsilon$ small we can ensure that the subsampled eigenvectors $\tilde x_j$ are nearly as good as $\hat x_j$. 
Without loss of generality, let $\tilde x$ be any eigenvector of $\tilde{A}$. 
We define the \textit{cross-validated eigenvalues} as 
$$\lambda\test(\tilde x)=\tilde x\T \tilde A\test \tilde x.$$
This estimator is a weighted sum of quantities from both split graphs that are independent and have identical population eigenvectors.

\subsection{The third piece: the cross-validated eigenvalues are asymptotically normal}


In this section, we provide a central limit theorem (CLT) for $\lambda\test(\tilde x)$.
We define the sample and population total variance as
$$\hat{\sigma}^2 = 2(\tilde x^2)\T \tilde A\test (\tilde x^2)-(\tilde x^2)\T \diag(\tilde A\test) (\tilde x^2)$$
and
$$\sigma^2 = 2(\tilde x^2)\T (\varepsilon P) (\tilde x^2) - (\tilde x^2)\T \diag(\varepsilon P) (\tilde x^2).$$
Next, we define the following delocalization condition on $\tilde x$:
\begin{equation}\label{eq:deloc-cond}
	\|\tilde x\|_\infty^2 = o(\sigma)
\end{equation}
where $\tilde{x}^2\in\R^n$ is the vector $x$ with all entries squared. 

\begin{theorem}[CLT for cross-validated eigenvalues] \label{thm:cltA}
	Assume that $\tilde x$ satisfies the delocalization condition \eqref{eq:deloc-cond}.
	Then,
	\begin{equation}\label{eq:clt}
		\frac{\lambda\test(\tilde x) - \lambda_{\varepsilon P}(\tilde x)}{\hat \sigma} \Rightarrow N(0,1).
	\end{equation}
\end{theorem}

The proof of Theorem \ref{thm:cltA} is in Supplementary Section \ref{sec:proof}. When $\tilde x$ is the eigenvector of $\tilde A$'s normalized form, a similar CLT result exists.

\begin{remark}[Sufficient condition]
Regarding the delocalization condition~\eqref{eq:deloc-cond}, when all entries of $\varepsilon P$ are of the same order $\rho = o(1)$, then  
$\sigma = \Theta(\rho^{1/2})$ and the condition boils down to $\|x\|_\infty = o(\rho^{1/4})$.
In Supplementary Section \ref{sec:inhomogeneous} (Corollary \ref{cor:hetero}), we discuss a sufficient condition for $\|x\|^2_\infty = o(\sigma)$ to hold in terms of the expected number of edges in $\tilde A\test$. 
\end{remark}

\begin{remark}[Bernoulli graphs]
Algorithm \ref{alg:eigcv} operates identically for Bernoulli graphs. While the theoretical justification requires conditioning on the edge-splitting pattern (Section \ref{sec:clt-bernoulli}), this conditioning is implicit in the randomization and does not affect the implementation.
\end{remark}

Based on Theorem \ref{thm:cltA}, we test the null hypothesis that the expected value of $\lambda\test(\tilde x)$ is zero,
$$H_0: \lambda_{\varepsilon P}(\tilde x) = 0 $$
Note that the null is equivalent to population-validate eigenvalue $\lambda_{P}(\tilde x)$ is zero, since $\varepsilon>0$ is a constant.
Under the null, $\tilde x$ is orthogonal to the latent space; in this case, we say that $\tilde x$ is not statistically useful (although, perhaps, it is still useful for tasks other than estimating graph dimension $k$).

\subsection{Extension to Bernoulli graphs} \label{sec:clt-bernoulli}

While the three key pieces presented above provide clean intuition and theory for Poisson graphs, extending to Bernoulli graphs requires a modified approach. The fundamental challenge is that Lemma 3.1, which guarantees the independence between $\tilde{A}$ and $\tilde{A}\test$ under edge splitting for Poisson graphs, has no counterpart for Bernoulli graphs; when $A_{ij} \in \{0,1\}$, an edge cannot simultaneously appear in both $\tilde{A}$ and $\tilde{A}\test$, creating a negative dependence.

To generate an ``independent copy'' in the Bernoulli case, we instead split possible edges; let $\Omega \subseteq \{(i,j): 1 \leq i < j \leq n\}$ be a random subset of all $\binom{n}{2}$ possible edges, where each pair is included in $\Omega$ independently with probability $\varepsilon$. This set $\Omega$ determines which edges are \emph{eligible} for the test set, regardless of whether they appear in the observed graph. Conditionally on $\Omega$, we define:
\begin{align}
[\tilde A\test^\Omega]_{ij}&\overset{ind.}{\sim}\text{Bernoulli}([P_\Omega]_{ij}),\\
[\tilde A^\Omega]_{ij}&\overset{ind.}{\sim}\text{Bernoulli}([P-P_\Omega)]_{ij}).
\end{align}
where $[P^\Omega]_{ij}=P_{ij}$ if edge $(i,j)\in\Omega$ otherwise 0. Given $\Omega$, $\tilde A\test^\Omega$ and $\tilde A^\Omega$ are conditionally independent. For example, if $\tilde x\in\R^n$ is a eigenvector of $\tilde A^\Omega$, then $\tilde x$ and $\tilde A\test^\Omega$ are conditionally independent given $\Omega$.

Crucially, Algorithm \ref{alg:edge_splitting} (\texttt{ES}) provides an efficient implementation of this conditioning. Rather than explicitly generating $\Omega$ for all $O(n^2)$ possible edges, \texttt{ES} only processes the observed edges in $A$, which is $O(m)$ where $m$ is the number of edges. The resulting split is distributionally equivalent to first sampling $\Omega$, then generating the conditional graphs. Thus, the computational efficiency of the algorithm remains unchanged.
Under this conditioning, we have the following result (the proof is in Supplementary Section \ref{sec:proof}):

\begin{theorem}[Conditional CLT for Bernoulli graphs] \label{thm:bernoulli-clt}
Assume $\max_{i,j} P_{ij} = o(1)$. Suppose that with probability 1, a unit vector $\tilde{x}$ satisfies the delocalization condition $\|\tilde{x}\|^2_{\infty} = o(\sigma_{\Omega})$, where the population total variance is
$$\sigma^2_{\Omega} = 2(\tilde{x}^2)\T(P_\Omega - P_\Omega^2)\tilde{x}^2 - (\tilde{x}^2)\T \diag(P_\Omega - P_\Omega^2)\tilde{x}^2$$ 
with $[P_\Omega^2]_{ij}=[P_\Omega]_{ij}^2$ for simplicity,
and that $\sigma_{\Omega} > 0$ eventually. Then, conditionally on $\Omega$,
$$\frac{\lambda^{\Omega}\test(\tilde{x}) - \lambda_{P_\Omega}\tilde{x}}{\hat{\sigma}_{\Omega}} \Rightarrow N(0,1), \quad \text{almost surely},$$
where $\lambda^{\Omega}\test(\tilde{x}) = \tilde{x}^T \tilde{A}^{\Omega}\test \tilde{x}$ and $\hat\sigma^2_\Omega = 2(\tilde x^2)\T \tilde A\test^\Omega \tilde x^2 - (\tilde x^2)\T \diag(\tilde A\test^\Omega) \tilde x^2$ is the sample variance estimate.
\end{theorem}

The variance formula $\sigma^2_{\Omega}$ includes terms $(1 - [P_\Omega]_{ij})$ reflecting the Bernoulli structure, in contrast to the Poisson case in Theorem \ref{thm:cltA}. Since $(1 - [P_\Omega]_{ij}) \leq 1$, the Poisson variance formula provides an upper bound, enabling conservative inference. 

This theoretical modification has two practical implications. First, Algorithm \ref{alg:eigcv} (\texttt{EigCV}) operates identically for both Poisson and Bernoulli graphs—the conditioning is implicit in the randomization. Second, the negative dependence makes the test conservative, particularly for denser graphs where more edges create more dependence. As demonstrated in Supplementary Section \ref{sec:poisson_vs_bernoulli}, this conservativeness is most pronounced in denser graphs, for eigenvectors beyond the true dimension $k$. Despite reduced power, this conservativeness maintains valid inference, making the method reliable for the Bernoulli graphs that arise in practice.

Our main consistency result (Theorem \ref{thm:consistency}) in the next section extends to the Bernoulli case with minor modifications (see Supplementary Remark \ref{remark:consistency_bernoulli}), where $\varepsilon P$ is replaced with the conditional expectation. Thus, our method provides both valid inference and consistent estimation of $k$ for Bernoulli graphs, with the primary cost being conservative p-values that may reduce power in dense graphs.

\section{Cross-validated eigenvalue estimation}

In this section, we use Theorem~\ref{thm:cltA} to test the null hypothesis $H_0: \lambda_P(\tilde x_j) = 0$, where $\tilde x_j$ is the $j$th eigenvector of $\tilde A$. 
Ideally, the first $k$ eigenvectors will provide large values of $\lambda\test(\tilde x_j)$ for $j=1,2,...,k$, while the following eigenvectors have $\lambda\test(\tilde x_j) \approx 0$ for $j>k$, whose Z-scores distributed normally around zero (with variance one). 
Section \ref{sec:alg} states the algorithm and Section \ref{sec:consistency} provides the main theoretical result, i.e., this estimation is consistent. 

\subsection{The algorithm}\label{sec:alg}

Algorithm \ref{alg:eigcv} revokes \texttt{ES} and calculated up to $\km$ cross-validated eigenvalues.
The algorithm reports a $p$-value for each eigenvector.  These are then used to estimate $k$. 
For simplicity, we describe the algorithm for an undirected graph with the adjacency matrix $A\in \N^{n \times n}$; see Remark \ref{remark:rectangle} for rectangular or asymmetric $A$.

\begin{algorithm}
    \setstretch{1.35}
	\DontPrintSemicolon
	\caption{Eigenvalue cross-validation}
	\label{alg:eigcv}
	\KwData{
        symmetric adjacency matrix $A \in \mathbb{N}^{n \times n}$, edge splitting probability $\varepsilon \in (0,1)$, maximum possible graph dimension $\km>1$}
    \textbf{Optional} \KwData{
        number of cross-validation repetitions $n_\text{cv}$ (default to 10), and significance level $\alpha$ (default to 0.05)}
	\textbf{Procedure} \texttt{EigCV}$(A,\varepsilon,\km, \alpha, n_\text{cv})$\textbf{:}\;
	\Indp
	1. \For{$i = 1, \dots, n_\text{cv}$}{
		a. $\tilde{A}$, $\tilde{A}\test\leftarrow$\texttt{ES}$(A,\varepsilon)$
		\tcp*{Algorithm \ref{alg:edge_splitting}} 
		b. (Optional) Compute graph Laplacian $\tilde L$ with $\tilde A$, 
            $$\tilde L = \tilde D \tilde A \tilde D,$$
            where $\tilde D\in\R^{n \times n}$ is a diagonal matrix with element $\tilde D_{ii} = (\tilde d_i + \tilde \tau)^{-1/2}$ and node degrees $\tilde d_i = \sum_j \tilde A_{ij}$ and regularization parameter $\tilde \tau = \sum_i \tilde d_i / n$.\;
		c. Compute the leading $\km$ eigenvectors of $\tilde L$ (or $\tilde A$), as $\tilde x_1, \dots, \tilde x_\km$. \;
		d. \For{$j = 2, \dots, \km$}{
            Compute the test statistic
			$$T_{f,j} = \frac{\lambda\test (\tilde x_j)}{\hat \sigma_j},$$
			where $\lambda\test(\tilde x) = \tilde x\T \tilde A\test \tilde x$, 
            and $\hat \sigma_j = \sqrt{2\varepsilon (\tilde x_j^2)\T {A} \tilde x_j^2-\varepsilon(\tilde x_j^2)\T \diag(A) \tilde x_j^2}$ is the standard error evaluated using the full graph,
            and $\tilde x_j^2$ is the vector $\tilde x_j$ with each element squared.\;}}
	2. \For{$j = 2, \dots, \km$}{ 
        Compute $T_j$ as the mean of the $T_{1,j}, \dots, T_{n_\text{cv},j}$ and compute the one-sided $p$-value $p_j = 1-\Phi(T_j)$, where $\Phi$ is the cumulative distribution function of the standard normal distribution.\;}
	\Indm
	\KwResult{The graph dimensionality estimate: $\argmin_{k\le\km}\{p_{k} \ge \alpha\} - 1$.\;}
\end{algorithm}

In Algorithm \ref{alg:eigcv}, we allow for several algorithmic modifications from the main statistical theory in Section \ref{sec:consistency} (Theorem \ref{thm:consistency}) because they are practically advantageous. These include repetition of cross-validation (Step 1), graph Laplacian (Step 1b), and a delocalization check. We discuss the choices in order.

\begin{remark}[Repetition of cross-validation]
Increasing $n_\text{cv}>1$ helps to remove the randomness in the $p$-values generated from edge splitting $\texttt{ES}$ (see Supplementary Section \ref{app:bernoullivspoisson10fold} for further discussion of $\textrm{folds} > 1$).  
\end{remark}

\begin{remark}[Graph Laplacian]
While the technical parts of the paper use the eigenvectors of an adjacency matrix, the proposed algorithm instead uses eigenvectors from the normalized and regularized adjacency matrices. 
Importantly, the normalized and regularized form of $\E(A)$ has the same number $k$ of non-zero eigenvalues as $\E(A)$.
For dense graphs, $L$ has similar statistical properties to $ A$. However, for sparse graphs, $L$ helps reduces localization of eigenvectors \citep{le_concentration_2017, zhang2018understanding}.
\end{remark}

\begin{remark}[Check for delocalization]
Finally, there is an additional step needed in Theorem \ref{thm:consistency} to check for delocalization. This technical requirement is further discussed in Section \ref{sec:consistency}. This step is not used in \texttt{EigCV} or in our implementation.  We do not include a check for delocalization because we find in the simulation in Supplementary Section \ref{sec:poisson_vs_bernoulli} that when an eigenvector $\tilde x_j$ delocalizes, $(\tilde x_j^2)\T {A} \tilde x_j^2$ in the formula for $\hat \sigma_j$ is very large, thus leading to conservative inferences.
\end{remark}


Algorithm \ref{alg:eigcv} (\texttt{EigCV}) operates identically for Bernoulli graphs. Furthermore, the framework of \texttt{EigCV} easily extends to two other settings, rectangular incidence matrices and a test of independence for contingency tables. We remark on the two extensions in order.

\begin{remark}[Rectangular incidence matrices] \label{remark:rectangle} 
If the matrix $A \in \N^{r \times c}$ is either rectangular or asymmetric (e.g., the adjacency matrix for a directed graph, the incidence matrix for a bipartite graph, a contingency table, etc.), then eigenvectors should be replaced by singular vectors. In Step 1c of \texttt{EigCV}, compute the singular vector pairs $\tilde u_j$ and $\tilde v_j$.  Then, the define test statistic in Step 1d as 
\[T_j = \frac{ \tilde u_j\T \tilde A\test \tilde v_j}{\sqrt{\varepsilon (\tilde u_j^2)\T A \tilde v_j}}.\]
Theorem \ref{thm:cltA} extends under analogous conditions to this setting. Our \texttt{R} package, \texttt{gdim} (available from CRAN), includes this extension. 
\end{remark}

\begin{remark}[Contingency tables] 
Suppose $A \in \mathbb{N}^{r \times c}$ is a contingency table with multinomial elements.  Note that the  $\chi^2$ test of independence tests the null hypothesis that
$\E(A)$ is rank 1, i.e., $k = 1$.  To test if $\E(A)$ has rank greater than $k=1$ and potentially reject independence, one can apply \texttt{EigCV} to $A$ with $k_{max} =2$, using the extension to rectangular matrices in Remark \ref{remark:rectangle}.  
The three key pieces for the cross-validation apply to this setting, thus enabling this approach. First, if the distribution of $A$ is multinomial, then \texttt{ES} provides two independent matrices.  Second, in expectation, those matrices have identical singular vectors and the same number of non-zero singular values. Third, the CLT in Theorem \ref{thm:cltA} extends to this data generating model with analogous conditions.
\texttt{EigCV} is potentially powerful for alternative hypotheses where $\E(A)$ has a large second eigenvalue. 
In contrast to the traditional Pearson's $\chi^2$ test for independence, \texttt{EigCV} handles a large number of rows and columns and a sparse $A$, where the vast majority of elements are zeros.  Moreover, it has the added advantage that the singular vectors $\tilde u_2$ and $\tilde v_2$ estimate where the deviation from independence occurs, thus making the results more interpretable.  
This is an area of our ongoing research.  
\end{remark}

\subsection{Statistical consistency}\label{sec:consistency}

This section states a consistency result for a modified version of the algorithm stated in Algorithm \ref{alg:eigcv-mod} in Supplementary Section \ref{sec:consistency-proof}.
The main modification is the addition of a delocalization test. We use $K$ for the true latent dimension and $\hat K$ for its estimate.

We will make some further assumptions.
Let $P=\rho_nP^0$, where $0 < \rho_n < 1$ controls the sparsity of the network, and $P^0=U\Lambda^0 U\T$ is a matrix of rank $k$ with $P^0_{ij} \leq 1$ for all $i,j$. Here, $\Lambda^0=\diag(\lambda_1^0,\cdots,\lambda^0_{K})$ is the diagonal matrix of its non-increasing eigenvalues, and $U=(u_1,\cdots,u_{K})$ contains the corresponding eigenvectors.
We first consider the signal strength in the population adjacency matrix. 
The magnitude of the leading eigenvalues characterize the useful signal in the data; only if they are sufficiently large is it possible to identify them from a finite graph sample. As such, the first assumption requires that the leading eigenvalues of the population graph are of sufficient and comparable magnitude. We also include necessary assumptions on the sparsity of the graph. 
\begin{assumption}[Signal strength and sparsity]\label{assump:matrix_SSNR}
	We assume that there exist positive constants $\psi_1,\psi_1'$ such that $$\kappa:={\lambda_1^0}/{\lambda_K^0}\in(0,\psi_1),\qquad \lambda_1^0\geq \psi_1' n.$$  In addition, we assume that $P^0_{ij} \leq 1$ for all $i,j$ and that the network sparsity satisfies $ c_0 \frac{\log^{\xi_0}n}{n}\leq \rho_n\leq c_0'n^{-\xi_1}$, for some constants $\xi_0>2$, $\xi_1\in(0,1)$, $c_0,c_0'>0$.
\end{assumption}
Observe that Assumption \ref{assump:matrix_SSNR} implies in particular that $\psi_1^{-1}\psi_1'n\rho_n\leq\lambda_K\leq \lambda_1\leq n\rho_n$ since $\lambda_1\leq \tr(P)\leq n\rho_n$. Assumption \ref{assump:matrix_SSNR} is less strict than the assumptions in \cite{li_network_2020}. This is because we do not require a minimum gap between distinct eigenvalues, which is hard to satisfy in practice.

Next, we consider a property of the population eigenvectors. The notion of coherence was previously introduced by \cite{candes2009exact}. 
Under the parametrization of Assumption~\ref{assump:matrix_SSNR}, the coherence of $U$ 
is defined as
\begin{align*}
	\mu(U)=\max_{i\in[n]}\frac{n}{K}\Vert U\T e_i\Vert^2=\frac{n}{K}\Vert U\Vert_{2,\infty}^2,
\end{align*}
where $e_i$ is the $i$-th standard basis vector.
A lower coherence indicates that the population eigenvectors are more spread-out---that is, they are not concentrated on a few coordinates.  
\begin{assumption}[Coherence]\label{assump:coherence}
	We assume $\mu(U)\leq \mu_0$, for some constant $\mu_0>1$.
\end{assumption}

Our main theoretical result asserts the consistency of our cross-validated eigenvalue estimator for estimating the latent dimension.
The proof of Theorem \ref{thm:consistency} is in Supplementary Section \ref{sec:proof}.
\begin{theorem}[Consistency]\label{thm:consistency}
	Suppose $A\in\R^{n \times n}$ is a Poisson graph satisfying Assumptions~\ref{assump:matrix_SSNR} and~\ref{assump:coherence}. Let $K$ be the true latent space dimension, and let $\hat K$ be the output of Algorithm~\ref{alg:eigcv-mod} (see Supplementary Section \ref{sec:consistency-proof}) with edge splitting probability $\varepsilon$. 
	Then,
	\begin{eqnarray*}
		\mathbb P\left(\hat K=K\right) \to 1 & \text{as} &  n\to \infty.
	\end{eqnarray*}
\end{theorem}

\section{Comparing to other approaches}

This section compares the proposed method (\texttt{EigCV}) with some existing graph dimensionality estimators using both simulated and real graph data. 
An \texttt{R} implementation of \texttt{EigCV} is available from the CRAN package \texttt{EigCV}. 
Throughout, we set the graph splitting probability $\varepsilon$ to 0.05, set the significance level cut-off at $\alpha=0.05$, and $n_\text{cv} = 10$.

We compare \texttt{EigCV} to (1) BHMC, a spectral method based on the Bethe-Hessian matrix with correction \citep{le_estimating_2019}; (2) LR, a likelihood ratio method adapting a Bayesian information criterion \citep{wang_likelihood-based_2017}; (3) ECV, an edge cross-validation method with an area under the curve criterion \citep{li_network_2020}; (4) NCV, a node cross-validation using an binomial deviance criterion \citep{chen_network_2018}; and (5) StGoF (with $\alpha=0.05$), a stepwise goodness-of-fit estimate \citep{jin_estimating_2020}. 
We performed all computations in \texttt{R}.
For (1)-(4), we invoked the \texttt{R} package \texttt{randnet}, and for (5), we implemented the original Matlab code (shared by the authors) in \texttt{R}.

\subsection{Numerical experiments} \label{sec:numerical_simulations}

This section presents several simulation studies that compare our 
method with other approaches to graph dimensionality. 
We sampled random graphs with $n=2,000$ nodes and $k=10$ blocks from the Degree-Corrected Stochastic Blockmodel (DCSBM). Specifically, for any $i,j=1,2,...,n$, $$A_{ij}\overset{\text{ind.}}{\sim}\text{Bernoulli}\left(\theta_i\theta_j B_{z(i)z(j)}\right),$$ where $z(i)\in\{1,2,...,k\}$ is the block membership of node $i$, and $B\in\R^{k \times k}$ is the block connectivity matrix, with $B_{ii}=0.28$ and $B_{ij}=0.08$ for $i,j=1,2,...,k$, and $\theta_i>0$ is the degree parameters of node $i$. We investigated the effects of 
degree heterogeneity by drawing $\theta_i$'s from three distributions (before scaling to unit sum): (i) a point mass distribution, (ii) an Exponential distribution with rate 5, (iii) a Pareto distribution with location parameter 0.5 and dispersion parameter 5. From (i) to (iii), the node degrees become more heterogeneous. Finally, to examine the effects of sparsity, we chose the expected average node degree in $\{25, 30, ..., 60\}$. For each simulation setting, we evaluated all methods 100 times. 

Figure \ref{fig:accuracy} displays the accuracy of all graph dimensionality methods. Here, the accuracy is the fraction of times an estimator successfully identified the true underlying graph dimensionality (which is 10).\footnote{Besides comparison of accuracy, we also compared the deviation of the estimation by each method, for which similar results hold consistently (see Supplementary Figure \ref{fig:dev}).} 
From the results, both BHMC and ECV offered satisfactory estimation when the graph is degree-homogeneous and the average degree becomes sufficiently large, while they were affected drastically by the existence of 
degree heterogeneity. 
The LR estimate was affected by degree heterogeneity as well (although less than BHMC or ECV) and also required a relatively large average node degree to estimate the graph dimensionality. 
The NCV methods failed to estimate the graph dimensionality under 
most settings. 
The StGoF estimate worked better for degree-heterogeneous graphs but required a larger average node degree for accuracy. 
It is also worth pointing out that the LR and StGoF methods tended to over-estimate the graph dimensionality when the average degree is large, especially for the power-law graphs (see Supplementary Figure \ref{fig:dev}). 
Finally, our method provided a much more accurate dimensionality estimate overall, requiring smaller average node degree and allowing degree heterogeneity. 
In addition, our testing approach also enjoys a strong advantage of reduced computational cost. To show this, Figure \ref{fig:time} depicts the average runtime for each method. 
It can be seen that the proposed method and BHMC are faster than competing methods by several orders of magnitude. The computational complexity of each StGoF iteration (or test) is at least $O(n^2)$, regardless of whether the graph is sparse or not. Consequently, StGoF requires the longest runtime.

\begin{figure}
	\centering
	\includegraphics[width=1\linewidth]{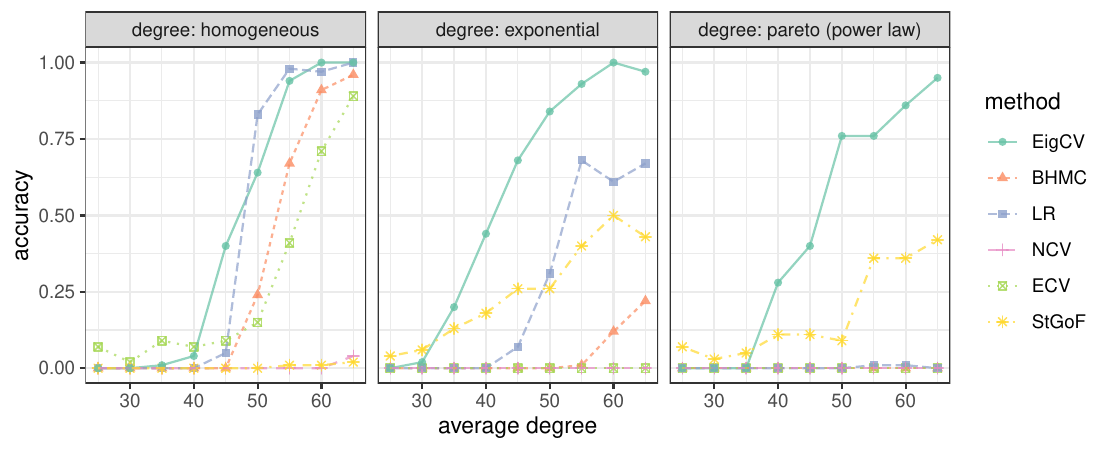}
	\caption{Comparison of accuracy for different graph dimensionality estimates under the DCSBM. The panel strips on the top indicate the node degree distribution used. Within each panel, each colored line depicts the relative error of each estimation method as the average node degree increases. Each point on the lines are averaged across 100 repeated experiments. }
	\label{fig:accuracy}
\end{figure}

\begin{figure}
	\centering
	\includegraphics[width=0.5\linewidth]{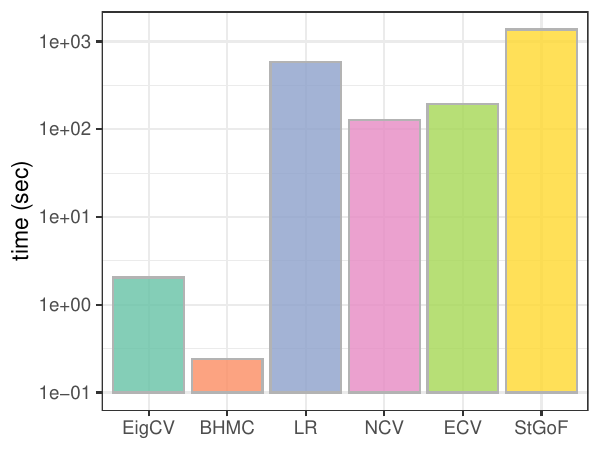}
	\caption{Comparison of runtime for the different graph dimensionality methods. Each colored bar indicates the runtime of applying each method on a DCSBM graph with 2000 nodes and 10 blocks. The maximum graph dimensionality is set to 15 for all methods. The runtime was averaged across 100 repeated experiments.}
	\label{fig:time}
\end{figure}

\subsection{Email network} \label{sec:email}
A real data network was generated using email data within a large European research institution, with each node representing one of the 1005 core members \citep{leskovec2007graph}. There is an edge from node $i$ to node $j$, if $i$ sent at least one email to $j$. 
The dataset also contains 42 ``ground-truth'' community memberships of the nodes. That is, each individual belongs to exactly one of 42 departments at the research institute. For simplicity, we removed the 14 small departments that consist of less than 10 members (see Supplementary Table \ref{tbl:email_42} for similar results without the filter). 
This resulted in a directed and unweighted network with a total of 936 nodes from 28 communities. 

We applied the graph dimensionality methods to estimate the number of clusters in the network. 
For the randomized methods (including ECV, NCV, and our proposed method), we ran them 25 times and report the mean and standard deviation of the estimates. 
For the methods that report a $p$-value (including StGoF and our proposed method), we use a significance level of $\alpha = 0.01$, followed by a multiplicity correction using the procedure of \cite{benjamini1995controlling}. 
We set $\km=50$.
Finally, we chose the splitting probability to be 0.05, as the network is sparse with an average node degree of 23.5. 
Table \ref{tbl:email} lists the inferences made by each method. As shown, our method provided an estimate that is close to the true number of departments within the institute.
BHMC, LR, NCV, and ECV all estimated small numbers of clusters, while StGoF went significantly larger ($\ge50$). These observations were consistent with the simulation results (see Supplementary Figure \ref{fig:dev}). Among all the others, only the proposed method provided a close estimate ($\approx28$) to the true number of departments. 
Similarly to the simulation results, the BHMC method and our method are more computationally efficient, with much shorter runtime than the others.

\begin{table}[ht]
	\centering
	\caption{Comparison of graph dimensionality estimates using the email network among members in a large European research institution. Each members belongs to one of 28 departments.}
	\label{tbl:email}
	\smallskip
	\begin{tabular}{lrr}
		\hline
		Method & Estimate (mean) & Runtime (second) \\ 
		\hline 
		EigCV & 28.3 & 0.68 \\ 
		BHMC & 14 & 0.02 \\ 
		LR & 17 & 85.19 \\ 
		NCV & 6.5 & 204.97 \\
		ECV & 16.5 & 41.07 \\ 
		StGoF & $\ge50$ & 397.47 \\
		\hline
	\end{tabular}
\end{table}

\section{Discussion}

In this paper, we proposed a new way to estimate the number of latent dimensions in a graph $k$ using the concept of cross-validated eigenvalues. Through edge splitting and thanks to a simple central limit theorem (or conditional CLT for Bernoulli graphs), the estimation of cross-validated eigenvalues is efficient for very large graphs. The paper also provides theoretical justification showing that the estimator is consistent. Our simulations and empirical data application validate the theory and further demonstrate the efficacy of the proposed method.

In addition to being quickly computable, a key advantage of cross-validated eigenvalues is our rigorous understanding of their behavior outside of the asymptotic setting where all $k$ dimensions can be estimated.  Theorem \ref{thm:cltA} encodes this rigorous understanding into a $p$-value. This theorem requires very little of the population matrix $P$; it does not presume that it is from a Degree-Corrected Stochastic Blockmodel, nor does it presume the actual rank of $P$ or the order of the eigenvalue being tested. Of course, this level of ease and generality comes with a price.  For Bernoulli graphs specifically, the inference is conservative rather than exact, though still valid.  In addition, we only get to compute the eigenvectors with $1-\varepsilon$ fraction of the edges.  The natural possibility is to estimate $k$ with a fraction of the edges and then recompute the eigenvectors with the full graph.  Going forward, we hope others will join us in crafting new estimators for $\lambda_P(\hat x_j)$ that do not require leaving out edges.

\section*{Acknowledgments}
This research is partially supported by the National Science Foundation under Grant DMS-1612456 (to K.R.), DMS-1916378 (to S.R. and K.R.), CCF-1740707 (TRIPODS Phase I) and DMS-2023239 (TRIPODS Phase II), and the Army Research Office under Grant W911NF-15-1-0423 (to K.R.). 
We thank S\"{u}nd\"{u}z Kele\c{s}, Po-Ling Loh, Michael A Newton, Kris Sankaran, Muzhe Zeng, Alex Hayes, E Auden Krauska, Sijia Fang, and Jitian Zhao for all the helpful discussions.

\clearpage
\appendix
\begin{center}
    {\LARGE\bf Supplementary Materials}
\end{center}
	
\beginsupplement

\spacingset{1.5} 

\section{Technical proofs} \label{sec:proof}

\subsection{Proof of Proposition \ref{prop:signal_strength}}

\begin{proof} 
Let $\hat X \in \mathbb{R}^{n \times q}$ contain the leading $q$ sample eigenvectors $\hat x_1, \dots, \hat x_q$ in its columns.  We aim to show that 
\[\min_{\Gamma\text{\normalfont~is~diagonal}} \| P - \hat X \Gamma \hat X\T \|_F\]
contains $\lambda_{P}(\hat x_1), \dots, \lambda_{P}(\hat x_q)$ down the diagonal.

Using the symmetry of $P$ and the cyclic property of the trace, we obtain
\begin{eqnarray*}
	\| P - \hat X \Gamma \hat X\T \|_F^2  &=& 
	\tr(P^2) + \tr(\hat X \Gamma^2 \hat X\T) - 2 \tr(P\hat X \Gamma \hat X\T) \\
	&=&\tr(P^2) + \tr(\Gamma^2) - 2 \tr(\hat X\T P \hat X \Gamma).
\end{eqnarray*}
Taking a derivative with respect to the diagonal of $\Gamma$ and setting equal to zero gives 
\[\Gamma = \mathrm{diag}(\hat X\T P \hat X)\]
which contains $\lambda_P( \hat x_1), \dots, \lambda_P( \hat x_q) $ down the diagonal.
\end{proof}

\subsection{Proof of asymptotic normality}
This section proves the two versions of the central limit theorem (CLT) and a sufficient condition described in Section \ref{sec:splitting}.

\subsubsection{CLT in the Poisson case}
\begin{proof}[Proof of Theorem \ref{thm:cltA}]
We first show the CLT with the population total variance $\sigma^2$ (in the denominator), 
\begin{equation}\label{eq:cltpopulation}
    \frac{\lambda\test(\tilde x) - \lambda_{\varepsilon P}(\tilde x)}
    {\sigma} \Rightarrow N(0,1).
\end{equation}
by verifying Lyapunov's condition. 
Then, we show the CLT holds when replacing population variance with the sample total variance $\hat\sigma$ via Slutsky’s Lemma. 

Since $\tilde A\test$ is symmetric, only the triangular elements are independent. That is, for all $i \le j$, $[\tilde A\test]_{ij}$ are independent. 
For any $1\le i\le j \le n$, define random variables $X_{ij}=x_ix_j([\tilde A\test]_{ij}-\varepsilon P_{ij})(2-\I(i=j))$, where $\I(\cdot)$ is the indicator function.
Then, $X_{ij}$ has mean zero and variance $\tilde x_i^2\tilde x_j^2\left(\varepsilon P_{ij}\right)(2-\I(i=j))^2$. 
Furthermore, the numerator in Equation \eqref{eq:cltpopulation} can be written as the sum of $X_{ij}$'s, $\lambda\test(\tilde x) - \lambda_{\varepsilon P}(\tilde x)=\sum_{i\le j}X_{ij}$. The sum of $X_{ij}$'s variances is equal to the definition of $\sigma^2$, 
$$\sum_{i \le j} \tilde x_i^2\tilde x_j^2 (\varepsilon P_{ij}) (2-\I(i=j))^2=2(\tilde x^2)\T (\varepsilon P) (\tilde x^2)- (\tilde x^2)\T \diag(\varepsilon P) (\tilde x^2)=\sigma^2.$$

We verify the fourth moment Lyapunov's condition on the summation of $X_{ij}$ for $i\le j$ (see, e.g.,~\citet[Exercise 3.4.12]{durrett_2019}). 
\begin{eqnarray*}
    &&\frac
    { \sum_{i \le j}  \E\left| \tilde x_i \tilde x_j ([\tilde A\test]_{ij} - \varepsilon P_{ij}) (2-\I(i=j))\right|^4}
    {\sigma^4} \\
    &\overset{\text{(i)}}{\leq}&
    \frac{ \sum_{i \le j} (\tilde x_i\tilde x_j)^4 (4\varepsilon P_{ij}) (2-\I(i=j))^4}
    {\sigma^4} \\
    &\le& 
    \frac{ 16 \|\tilde x\|_\infty^4 \sum_{i \le j} \tilde x_i^2\tilde x_j^2 (\varepsilon P_{ij}) (2-\I(i=j))^2}
    {\sigma^4} \\
    &=& 
    \frac{ 16 \|\tilde x\|_\infty^4 }
    {\sigma^2} \\
    &\overset{\text{(ii)}}{=}& o(1), 
\end{eqnarray*}
where (i) comes from the fourth central moment of Poisson variable, $\E([\tilde A\test]_{ij}-\varepsilon P_{ij})^4=\varepsilon P_{ij}(1 + 3 \varepsilon P_{ij}) \leq 4 \varepsilon P_{ij}$, combining the assumption $\varepsilon P_{ij} \leq 1$, and (ii) is due to the delocalization condition. 
This shows Equation \eqref{eq:cltpopulation}.

Via Slutsky's Lemma, we can multiply the ratio in Equation \eqref{eq:cltpopulation} by any sequence that converges to one in probability and the result still holds. The proof is then concluded by showing that $\sigma/\hat \sigma$ converges to one in probability.  
Indeed, we have
\begin{eqnarray*}
    \Var\left(\frac{\hat \sigma^2}{\sigma^2}\right) &=& \frac{\Var(\hat \sigma^2)}{\sigma^2} \\
    &=& \frac{\Var\left(2 (\tilde x^2)\T \tilde A\test (\tilde x^2) - (\tilde x^2)\T \diag(\tilde A\test) (\tilde x^2) \right)}{\sigma^4} \\
    &=& \frac{\Var\left(\sum_{i \le j}\tilde x_i^2\tilde x_j^2[\tilde A\test]_{ij}(2-\I(i=j)^2\right)}{\sigma^4}\\
    &\overset{\text{(i)}}{=}& \frac{\sum_{i \le j}\tilde x_i^4\tilde x_j^4\varepsilon P_{ij}(2-\I(i=j))^4}{\sigma^4}\\
    &\le& \frac{4\|\tilde x\|_\infty^4\sum_{i \le j}\tilde x_i^2\tilde x_j^2\varepsilon P_{ij}(2-\I(i=j))^2}{\sigma^4}\\
    &=& \frac{4\|\tilde x\|_\infty^4}{\sigma^2}\\
    &\overset{\text{(ii)}}{=}& o(1),
\end{eqnarray*}
where (i) is due to the mutual independence among upper triangular elements of $\tilde A\test$, and (ii) is due to the delocalization condition.
So, by Chebyshev's inequality, $\hat \sigma^2/\sigma^2$ converges in probability to its expectation.  Note that
$\E (\hat \sigma^2 / \sigma^2)=1$ and that taking the inverse and the square root is continuous transformation.  So, the ratio
$\sigma/\hat \sigma$ 
converges in probability to one. 
\end{proof}

\subsubsection{A sufficient condition of CLT in the Poisson case} \label{sec:inhomogeneous}

The following corollary gives a sufficient condition for $\|\tilde x\|^2_\infty = o(\sigma)$ to hold in terms of $m$ and the expected number of edges in $\tilde A\test$. 
\begin{corollary} \label{cor:hetero}
	Using the setting of Theorem \ref{thm:cltA},  
	let $\pi \in \R^n$ be a probability distribution on the nodes with $\pi_i$ proportional to a node's expected degree. Define $\langle \pi, x^2 \rangle$ be the expected value of $x_I^2$ for $I$ drawn from $\pi$ and define $m= \sum_{i} d_{i}/2$ as the expected total number of edges.
	If $P$ is positive semi-definite and  
	\begin{equation*} \label{eq:degcor}
		\frac{ \|\tilde x\|^2_\infty}{\langle \pi, \tilde x^2 \rangle} = o\left(\sqrt{m}\right), 
	\end{equation*}
	then the CLT in Equation \eqref{eq:clt} holds. 
\end{corollary}

\begin{proof}
	The proof of Corollary \ref{cor:hetero} follows directly from the next lemma.  
	
	\begin{lemma}\label{lem:lb_sigma}
		Suppose $Q\in\R^{n\times n}_+$ is positive semi-definite. 
		Define $d = Q \1_n \in \R^n$ to be the expected degrees of the nodes $1, \dots, n$, where $\1_n\in\R^n$ is a vector of 1's. 
		Then,
		\[\sigma^2 = 2(\tilde x^2)\T Q \tilde x^2-(\tilde x^2)\T \diag(Q) \tilde x^2 \ge  \frac{\langle d, \tilde x^2 \rangle^2} {\sum_i d_i}.\]
	\end{lemma}
	\begin{proof}
		Define $y= \tilde x^2, \theta = d^{1/2}, \Theta = \diag(\theta) \in \R^{n \times n}, y_\theta = \Theta y$, and $\mathscr{L} = \Theta^{-1} Q \Theta^{-1}$.  
		Because the elements of $\theta$ are non-negative, $\mathscr{L}$ is non-negative definite.  
		
		The first part of the proof is to show that $\mathscr{L} \theta = \theta$.  This is because $\Theta^{-2} Q$ is a Markov transition matrix. So,
		\[\Theta^{-2} Q \1_n = \1_n \implies \Theta^{-1} Q \Theta^{-1} \Theta \1_n = \Theta \1_n\]
		and this implies that $\mathscr{L} \theta = \theta$. So, by the Perron-Frobenius Theorem, $\theta$ is the leading eigenvector of $\mathscr{L}$ with eigenvalue 1. 
		
		Let $\mathscr{L}$ have eigenvectors and eigenvalues $(\phi_1, \lambda_1), \dots, (\phi_n, \lambda_n)$,
		where $\phi_1 = \theta / \|\theta\|_2$, $\lambda_1=1$ and $0 \leq \lambda_j \leq 1$ for $j \neq 1$.
		Then,
		\[y\T Qy =  y_\theta\T \mathscr{L} y_\theta = \sum_{\ell=1}^n \lambda_\ell \langle \phi_\ell, y_\theta \rangle^2.\] 
		Keeping only the first order term on the right-hand side, we have
		$$y\T Qy \ge \lambda_1 \langle \phi_1, y_\theta \rangle^2 = \frac{\langle d, \tilde x^2\rangle^2}{\sum_i d_i}.$$
		The desired result follows from the fact that $\sigma^2=y\T Qy+y\T (Q-\diag (Q))y\geq y\T Qy$, since $y$ and $Q$ have non-negative entries.
	\end{proof}
	Applying the bound in the lemma to the delocalization condition and rearranging gives the claim.
\end{proof}

\subsubsection{Conditional CLT in the Bernoulli case}\label{}

In Bernoulli graphs, the conditional mean of $\lambda\test(\tilde x)$ is no longer $\lambda_{\varepsilon P}(\tilde x)$ but takes a conditional form.
Let $\Omega\in\{0,1\}^{n\times n}$ be the adjacency matrix of an Erd\"os--R\'enyi graph with edge probability $\varepsilon$. 
Denote as $\circ$ the element-wise matrix multiplication, 
$[P\circ\Omega]_{ij}:= P_{ij} \Omega_{ij},$
and let $J$ be the matrix filled with 1. 
We define the following random variables: 
$$[\tilde A\test^\Omega]_{ij}\overset{ind.}{\sim}\text{Bernoulli}([P\circ\Omega]_{ij}),$$
$$[\tilde A^\Omega]_{ij}\overset{ind.}{\sim}\text{Bernoulli}([P\circ(J-\Omega)]_{ij}).$$
Observe that the random matrices $\tilde A^\Omega$ and $\tilde A\test^\Omega$ are conditionally independent given $\Omega$. 
Let $\tilde x\in\R^n$ be any vector that only depends on $\tilde A^\Omega$ (e.g., an eigenvector). It follows that $\tilde x$ and $\tilde A\test^\Omega$ are conditionally independent given $\Omega$ as well.

Given any unit vector $\tilde x\in\R^n$, we define the test statistic and sample total variance 
which will appear in the conditional CLT:
$$\lambda\test^\Omega(\tilde x)= \tilde x\T \tilde A\test^\Omega \tilde x,$$
$$\hat\sigma^2_\Omega = 2(\tilde x^2)\T \tilde A\test^\Omega \tilde x^2 - (\tilde x^2)\T \diag(\tilde A\test^\Omega) \tilde x^2.$$
The conditional expectation of $\lambda^\Omega\test(\tilde x)$ given $\Omega$ is $\lambda_{P\circ\Omega}(\tilde x)=\tilde x\T (P\circ\Omega) \tilde x$. Below, $\tilde x$ and $\Omega$ are technically sequences depending on $n$, but we do not explicitly indicate the dependence on $n$ for notational simplicity.


\begin{proof}[Proof of Theorem \ref{thm:bernoulli-clt}]
We first show the CLT with the population variance $\sigma^2_\Omega$ (in the denominator), 
\begin{equation}\label{eq:cltpopbern}
    \frac{\lambda\test^\Omega(\tilde x) - \lambda_{P\circ\Omega}(\tilde x)}
    {\sigma_\Omega} \Rightarrow N(0,1)
\end{equation}
by verifying Lyapunov's condition. 
Then, we show the CLT holds when replacing population variance with the sample variance $\hat\sigma_\Omega$ via Slutsky’s Lemma. 

Since $\tilde A\test^\Omega$ is symmetric, only the triangular elements are independent. That is, for all $i \le j$, $[\tilde A\test^\Omega]_{ij}$ are independent. 
For any $1\le i\le j \le n$, define random variables $X_{ij}=x_ix_j([\tilde A\test^\Omega]_{ij}-[P\circ\Omega]_{ij})(2-\I(i=j))$.
Then, $X_{ij}$ has conditional mean zero and conditional variance ${\tilde x_i^2\tilde x_j^2[P\circ\Omega]_{ij}\left(1-[P\circ\Omega]_{ij}\right)}(2-\I(i=j))^2$. 
Furthermore, the numerator in Equation \eqref{eq:cltpopbern} can be written as the sum of the $X_{ij}$'s, i.e, $\lambda\test^\Omega(\tilde x) - \lambda_{P\circ\Omega}(\tilde x) = \sum_{i\le j}X_{ij}$. The sum of the $X_{ij}$'s conditional variances is equal to $\sigma_\Omega^2$ by definition. 

We verify the fourth moment Lyapunov's condition on the summation of $X_{ij}$ for $i \le j$ (see, e.g.,~\citet[Exercise 3.4.12]{durrett_2019})
\begin{eqnarray*}
    \frac
    { \sum_{i \le j}  \E\left[\left|X_{ij}\right|^4\mid \Omega \right]}
    {\sigma^4_\Omega} &=& \frac
    { \sum_{i \le j}  \E\left[\left( x_ix_j([\tilde A\test^\Omega]_{ij}-[P\circ\Omega]_{ij})(2-\I(i=j))\right)^4\mid \Omega \right]}
    {\sigma^4_\Omega} \\
    &\overset{\text{(i)}}{\leq}&
    \frac{ \sum_{i \le j} (\tilde x_i\tilde x_j)^4 [P\circ\Omega]_{ij}(1-[P\circ\Omega]_{ij}) (2-\I(i=j))^4}
    {\sigma^4_\Omega} \\
    &\le& 
    \frac{4 \|\tilde x\|_\infty^4 \sum_{i \le j} \tilde x_i^2\tilde x_j^2 [P\circ\Omega]_{ij}(1-[P\circ\Omega]_{ij}) (2-\I(i=j))^2}
    {\sigma^4_\Omega} \\
    &=& 
    \frac{4 \|\tilde x\|_\infty^4 }
    {\sigma^2_\Omega} \\
    &\overset{\text{(ii)}}{=}& o(1), 
\end{eqnarray*}
where (i) comes from that the fourth central moment of a Bernoulli variable, $X\sim\text{Bernoulli}(p)$, satisfies $\E(X-p)^4= p(1-p)(p^3 + (1-p)^3) \leq p(1-p)$ 
and (ii) is due to the delocalization condition assumed. 
We also used that $\sigma_\Omega > 0$ eventually. This shows Equation \eqref{eq:cltpopbern}.

Via Slutsky's Lemma, we can multiply the ratio in Equation \eqref{eq:cltpopbern} by any sequence that converges to one in probability and the result still holds. We then show that $\sigma_\Omega/\hat \sigma_\Omega$ converges to one in probability.  
Indeed, we have 
\begin{eqnarray*}
    \Var\left(\frac{\hat \sigma^2_\Omega}{\sigma^2_\Omega}\mid \Omega \right) &=& \frac{\Var\left(\sum_{i \le j}\tilde x_i^2\tilde x_j^2[\tilde A\test^\Omega]_{ij}(2-\I(i=j))^2\mid \Omega\right)}{\sigma^4_\Omega}\\
    &=& \frac{\sum_{i \le j}\tilde x_i^4\tilde x_j^4[P\circ\Omega]_{ij}(1-[P\circ\Omega]_{ij})(2-\I(i=j))^4}{\sigma^4_\Omega}\\
    &\le& \frac{4\|\tilde x\|_\infty^4\sum_{i \le j}\tilde x_i^2\tilde x_j^2[P\circ\Omega]_{ij}(1-[P\circ\Omega]_{ij})(2-\I(i=j))^2}{\sigma^4_\Omega}\\
    &=& \frac{4\|\tilde x\|_\infty^4}{\sigma^2_\Omega}\\
    &\overset{\text{(i)}}{=}& o(1)
\end{eqnarray*}
where (i) is the delocalization assumption.
So, by Chebyshev's inequality, $\hat \sigma^2_\Omega/\sigma^2_\Omega$ converges in probability to its expectation given $\Omega$.  Note that
\begin{eqnarray*}
\E \left(\frac{\hat \sigma^2_\Omega}{ \sigma^2_\Omega}\mid \Omega\right) &=& 
\frac{
2(\tilde x^2)\T 
(P\circ\Omega) \tilde x^2 - (\tilde x^2)\T \diag(P\circ\Omega) \tilde x^2}{2(\tilde x^2)\T 
(P\circ\Omega-P\circ P\circ\Omega) \tilde x^2 - (\tilde x^2)\T \diag(P\circ\Omega-P\circ P\circ\Omega) \tilde x^2}\\
&\le&\frac{1}{1-\max_{i,j} P_{ij}}\\
&\overset{\text{(i)}}{=}&\frac{1}{1-o(
1
)}
\end{eqnarray*}
where (i) is due to the graph sparsity assumption. 
Similarly $\E \left(\frac{\hat \sigma^2_\Omega}{ \sigma^2_\Omega}\mid \Omega\right) \geq 1$.
Since taking the inverse and the square root is continuous transformation.  So, the ratio
$\sigma_\Omega/\hat \sigma_\Omega$ 
converges in probability to one, conditionally on $\Omega$. 
\end{proof}

\subsection{Proof of consistency}
\label{sec:consistency-proof}

This section details the proof of Theorem \ref{thm:consistency}. 

\paragraph{Notation}
We use the notation $[n]$ to refer to $\{1,2,...,n\}$. For any real numbers $a, b \in \mathbb R$, we denote $a \lor b = \max\{a, b\}$ and $a \land b = \min\{a, b\}.$ For non-negative $a_n$ and $b_n$ that depend on $n$, we write $a_n \lesssim b_n$ to mean $a_n \leq C b_n$ for some constant $C>0$, and similarly for $a_n \gtrsim b_n$. Also we write $a_n =O(b_n)$ to mean $a_n \leq C b_n$ for some constant $C>0$. The matrix spectral norm is
$\|M\| = \max_{\|x\|_2 = 1} \|M x\|_2$, the matrix max-norm is
$\|M\|_{\max} = \max_{i,j} |M_{ij}|$, 
and the matrix $2 \to \infty$ norm is
$\|M\|_{2, \infty} 
= \max_i \|M_{i,\cdot}\|_2$. 

\subsubsection{Modified algorithm}

Algorithm~\ref{alg:eigcv-mod} is used in the consistency result.
\begin{algorithm}
    \setstretch{1.35}
	\DontPrintSemicolon
	\caption{Modified eigenvalue cross-validation}
	\label{alg:eigcv-mod}
	\KwData{Adjacency matrix $A \in \mathbb{N}^{n \times n}$, edge splitting probability $\varepsilon \in (0,1)$\;}
	\textbf{Procedure} \texttt{EigCV'}$(A,\varepsilon,\km)$\textbf{:}\;
	\Indp
	1. Obtain $\tilde{A}$, $\tilde{A}\test\leftarrow$\texttt{ES}$(A,\varepsilon)$ from splitting $A$ and set $S = \emptyset$. \tcp*{Algorithm \ref{alg:edge_splitting}}
	
	2. \For{$k = 2, \dots, \km$}{
		a - compute $\tilde\lambda\test (\tilde x_k) = \tilde x_k\T \tilde A\test \tilde x_k$ and $\tilde \sigma_k = \sqrt{\frac{\varepsilon}{1-\varepsilon} (\tilde x_k^2)\T \left(2\tilde A-\diag(\tilde A)\right) \tilde x_k^2}$\;
		b - if 
		$$
		\norm{\tilde{x}_k}^2_\infty\leq\min\left\{\frac{\tilde\sigma^2_k}{\log^2 n}, \frac{\log n}{n}\right\},
		$$
		add $k$ to $S$ and compute
		$$T_k =  \frac{\tilde \lambda\test (\tilde x_k)}{\tilde \sigma_k}.$$
	}
	\Indm
	\KwResult{The graph dimensionality estimate: $\hat{K} = |\{T_k \geq \sqrt{n \log n} : k \in S\}|$. }
\end{algorithm}

\subsubsection{Some concentration bounds}
\label{sec:proof_prelim}

We will need several concentration bounds for Poisson random variables. We derive them from standard results. 

We begin with a simple moment growth bound.
\begin{lemma}[Poisson moment growth]
	\label{lem:poisson-growth}
	Let $Z$ be a Poisson random variable with mean $\mu \leq 1$. There exists a universal constant $C > 0$ such that, for all integers $p \geq 2$,
	$$
	\E[|Z-\mu|^p] \leq C \mu \, \frac{p!}{2} \left(\frac{e}{2}\right)^{p-2}.
	$$
\end{lemma}
\begin{proof}
	We show that 
	\begin{equation}\label{eq:moment-growth}
		\E[|Z-\mu|^p] \leq C' \mu \, \left(\frac{p}{2}\right)^p,
	\end{equation}
	for some constant $C' > 0$. The claim then follows
	from Stirling's formula in the form
	$$
	\sqrt{2 \pi} p^{p+1/2} e^{-p} \leq p!,\qquad \forall p \geq 1.
	$$
	By the definition of the Poisson distribution
	and using the fact that $0 \leq \mu \leq 1$ by assumption, we have
	\begin{eqnarray*}
		\E[|Z-\mu|^p]
		&=& \sum_{z \geq 0} |z - \mu|^p e^{-\mu} \frac{\mu^z}{z!}\\
		&=& |\mu|^p e^{-\mu} + |1-\mu|^p e^{-\mu} \mu +  \sum_{z \geq 2} |z - \mu|^p e^{-\mu} \frac{\mu^z}{z!}\\
		&\leq& 2 \mu + \mu^2 e \left\{\sum_{z \geq 0} z^p \frac{e^{-1}}{z!}\right\}.
	\end{eqnarray*}
	The term in curly brackets on the last line is the $p$-th moment of a Poisson random variable with mean 1, which is $\leq C'' \left(\frac{p}{2}\right)^p$ for some constant $C'' > 0$ by \citet[Theorem 1]{ahle_sharp_2021}. Eq.~\eqref{eq:moment-growth} follows.
\end{proof}

The moment growth bound implies 
concentration for linear combinations
of independent Poisson random variables. 
\begin{lemma}[General Bernstein for Poisson variables]\label{lem:bernstein-poisson}
	Let $Z_1,\ldots,Z_m$ be independent
	Poisson random variables with respective means $\mu_1,\ldots,\mu_m \leq 1$. For any $\balpha = (\alpha_1,\ldots,\alpha_m) \in \R^m$ and
	$t > 0$,
	$$
	\Prob\left[\sum_{i=1}^m \alpha_i (Z_i - \mu_i) \geq t\right]
	\leq \exp\left(- \frac{t^2}{C' \mumax \|\balpha\|_2^2 + C'' \|\balpha\|_\infty t}\right),
	$$
	where $\mumax = \max_i \mu_i$ and $C', C'' > 0$ are universal constants.
\end{lemma}
\begin{proof}
	We use~\citet[Corollary 2.11]{boucheron_concentration_2013}. Observe that
	$$
	\sum_{i=1}^m \E[\alpha_i (Z_i - \mu_i)^2]
	= \sum_{i=1}^m \alpha_i^2 \mu_i
	\leq \mumax \|\balpha\|_2^2.
	$$
	Moreover, by Lemma~\ref{lem:poisson-growth} and Stirling's formula, 
	\begin{eqnarray*}
		\sum_{i=1}^m \E[\alpha_i^p(Z_i - \mu_i)_+^p]
		&\leq& \sum_{i=1}^m \alpha_i^p C \mu_i \, \frac{p!}{2} \left(\frac{e}{2}\right)^{p-2}\\
		&\leq&  C \mumax \|\balpha\|_2^2 \frac{p!}{2} \left(\frac{e}{2} \|\balpha\|_\infty\right)^{p-2}\\
		&\leq& \frac{p!}{2} v
		\left(\frac{e}{2} \|\balpha\|_\infty\right)^{p-2},
	\end{eqnarray*}
	where we define
	$$
	v := \max\left\{1,C \right\} \mumax \|\balpha\|_2^2.
	$$
	The claim then follows from ~\citet[Corollary 2.11]{boucheron_concentration_2013}.
\end{proof}

The moment growth bound also implies spectral norm concentration.
\begin{lemma}[Spectral norm of Poisson graph]
	\label{lem:concentration_p}
	Suppose $B\in\mathbb R^{n\times n}$ is the adjacency matrix of a Poisson graph with mean matrix $Q$ satisfying $Q_{ij} \leq 1$ for all $i,j$. 
	Let $q_{\max} = \max_{ij} Q_{ij}$
	and assume that $n q_{\max} \geq c_0\log^{\xi_0} n$ for some $\xi_0> 2$.
	Then, for any $\delta > 0$, there exists a constant $C''' >0$ such that
	\begin{align*}
		\| B-Q \| \leq C''' \sqrt{n q_{\max} \log n},
	\end{align*}
	with probability at least $1-n^{-\delta}$.
\end{lemma}
\begin{proof}
	We use~\citet[Theorem 6.2]{tropp2012user}.
	We first rewrite the matrix as a finite sum of independent symmetric random matrices
	\begin{align*}
		B-Q=\sum_{i=1}^n\sum_{j=i}^n(B_{ij}-Q_{ij})E^{i,j},
	\end{align*}
	where $E^{i,j}\in \mathbb R^{n\times n}$ with $E^{i,j}_{ij}=E^{i,j}_{ji}=1$ and $0$ elsewhere. 
	
	Observe that, for $i \neq j$,
	\begin{align*}
		(E^{i,j})^p =
		\left\{
		\begin{array}{ll}
			E^{i,i}+E^{j,j}  & \mbox{if }p= 2,4,\ldots \\
			E^{i,j} & \mbox{if } p=3,5,\ldots
		\end{array}
		\right.
	\end{align*}
	while, if $i = j$,
	\begin{align*}
		(E^{i,i})^p = E^{i,i}, \qquad p \geq 2.
	\end{align*}
	
	Let $X^{i,j}:=(B_{ij}-Q_{ij})E^{i,j}$. 
	Then $\mathbb EX^{i,j}=0$. 
	Moreover, for $i \neq j$ and $p = 2, 4, \ldots$, we have 
	$$
	\mathbb E (X^{i,j})^p
	= 
	\E(B_{ij}-Q_{ij})^p \, (E^{i,i}+E^{j,j})
	\preceq 
	C q_{\max} 
	\frac{p!}{2}
	\left(\frac{e}{2}\right)^{p-2}(E^{i,i}+E^{j,j}),
	$$
	by Lemma~\ref{lem:poisson-growth}.
	Similarly, for $i \neq j$ and $p = 3, 5, \ldots$,
	$$
	\mathbb E (X^{i,j})^p
	= 
	\E(B_{ij}-Q_{ij})^p \, E^{i,j}
	\preceq 
	C q_{\max} 
	\frac{p!}{2}
	\left(\frac{e}{2}\right)^{p-2}(E^{i,i}+E^{j,j}),
	$$
	where we used the fact that the matrix $(\begin{smallmatrix}1 & \alpha\\ \alpha & 1\end{smallmatrix})$ has eigenvalues $1 + \alpha, 1- \alpha \geq 0$ when $|\alpha| \leq 1$. When $i = j$, 
	$$
	\mathbb E (X^{i,i})^p
	= 
	\E(B_{ii}-Q_{ii})^p \, E^{i,i}
	\preceq 
	C q_{\max} 
	\frac{p!}{2}
	\left(\frac{e}{2}\right)^{p-2}(2 E^{i,i}).
	$$
	
	Define
	$$
	(\Sigma^2)^{i,j}
	:= C q_{\max} 
	(E^{i,i}+E^{j,j}).
	$$
	and
	\begin{eqnarray*}
		\sigma^2
		&=&
		\norm{\sum_{i=1}^n\sum_{j=i}^n (\Sigma^2)^{i,j}}
		=
		\norm{ C q_{\max} \sum_{i=1}^n\sum_{j=i}^n(E^{i,i}+E^{j,j})}
		\leq 2 C q_{\max} n,
	\end{eqnarray*}
	where the inequality holds since 
	$\sum_{i=1}^n\sum_{j=i}^n (E^{i,i}+E^{j,j})$
	is a diagonal matrix with maximum entry $2n$. 
	Then, 
	by~\citet[Theorem 6.2]{tropp2012user},
	\begin{eqnarray*}
		\Prob\left[
		\norm{B - Q} \geq t
		\right]
		&=&
		\Prob\left[
		\norm{\sum_{i=1}^n\sum_{j=i}^n X^{i,j}} \geq t
		\right]\\
		&\leq& n
		\exp\left(\frac{-t^2/2}{\sigma^2 + (e/2) t}\right)\\
		&\leq& n
		\exp\left(\frac{-t^2/2}{2 C q_{\max} n + (e/2) t}\right).
	\end{eqnarray*}
	Taking $t = C''' \sqrt{n q_{\max} \log n}$ and using the fact that
	$nq_{\max}\geq c_0\log^{\xi_0} n$, $\xi_0> 2$, gives the result.
\end{proof}

\subsubsection{Key properties of sample eigenvectors}

Consider the adjacency matrix $A$ of a Poisson
graph satisfying Assumptions~\ref{assump:matrix_SSNR} and~\ref{assump:coherence}.
Fixing $\varepsilon \in (0,1)$, let $\tilde A$ and $\tilde A\test$ be as in Section~\ref{sec:splitting}.
Let 
$P
= \rho_n P^0
=\mathbb E A=\sum_{j=1}^K\lambda_j x_j\T x_j$
with $\lambda_1\geq \lambda_2\geq\cdots\geq \lambda_K>0$. Let $\{\tilde{ x}_l\}_{l=1}^{\km}$ be the collection of eigenvectors associated with eigenvalues $\{\tilde\lambda_l\}_{l=1}^{\km}$ of $\tilde A$. Without loss of generality, we assume $\tilde\lambda_1\geq\tilde\lambda_2\geq\cdots\geq\tilde\lambda_{\km}$. 
Define 
\begin{equation}
	\label{eq:def-uhat-u}
	\hat U=(\tilde{ x}_1,\cdots,\tilde{ x}_K) \qquad \text{and} \qquad U=({x}_1,\cdots,{x}_K)\in\mathbb R^{n\times K}.
\end{equation}
We will need the following event:
\begin{align*}
	&\mathcal E^0=\left\{\norm{\tilde A-(1-\varepsilon)P}\leq C'''\sqrt{n\rho_n\log n}\right\}.
\end{align*}
Applying Lemma~\ref{lem:concentration_p} with $B := \tilde{A}$ and $Q := (1-\varepsilon) P$ shows that
$\mathcal E^0$ holds with high probability.

\paragraph{Concentration of signal eigenspace}
First, we use a version of the Davis-Kahan theorem to
show that the signal sample eigenvectors are close to
the signal population eigenspace.
\begin{lemma}[Signal eigenspace]
	\label{lem:eignvector_2}
	Under 
	event $\mathcal E^0$, there exists an orthonormal matrix $O\in \mathbb R^{K\times K}$ such that, for all $k\in[K]$,
	\begin{align*}
		\Vert\tilde{y}_k-x_k\Vert_2=O\left(\sqrt{\frac{\log n}{n\rho_n}}\right), \qquad \Vert\tilde{x}_k-y_k\Vert_2=O\left(\sqrt{\frac{\log n}{n\rho_n}}\right),
	\end{align*}
	where 
	$$
	\tilde{y}_l=\left(\hat UO\right)_{\cdot l}=\left(\sum_{i=1}^K(\tilde{x}_i)_jO_{il}\right)_{j=1}^n,\quad y_l=\left(UO\T\right)_{\cdot l}=\left(\sum_{i=1}^K(x_i)_jO_{li}\right)_{j=1}^n.
	$$
	Moreover, for all $k\in[K]$, $s\in[K]$, and $t\in[\km]\setminus[K]$,
	\begin{align*}
		&\langle x_s,y_k\rangle=O_{ks},\qquad \langle\tilde{x}_t,\tilde {y}_k\rangle=0.
	\end{align*}
\end{lemma}
\begin{proof}
	We use the variant of the Davis-Kahan theorem in~\citet[Theorem 2]{yu2015useful}.
	Under $\mathcal E^0$, $\norm{\tilde A-(1-\varepsilon)P}= O(\sqrt{n\rho_n\log n})$. By~\citet[Theorem 2]{yu2015useful},
	there exists an orthonormal matrix $O\in \mathbb R^{K\times K}$ such that, for all $l\in[K]$,
	$$
	\Vert\tilde{y}_l-x_l\Vert_2\leq \Vert\hat UO-U\Vert_{\F}
	=O\left( \frac{\Vert\tilde A-(1-\varepsilon)P\Vert}{\lambda_K}\right)
	=O\left(\sqrt{\frac{\log n}{n\rho_n}}\right),$$
	and
	$$\Vert\tilde{x}_l-y_l\Vert_2\leq \Vert\hat U-UO\T\Vert_{\F}=\Vert(\hat UO-U)O\T\Vert_{\F}=\Vert\hat UO-U\Vert_{\F}=O\left(\sqrt{\frac{\log n}{n\rho_n}}\right),$$
	where we 
	used $\lambda_K\geq \psi_1^{-1}\psi_1'n\rho_n$, which holds under Assumption~\ref{assump:matrix_SSNR}.
	
	By the orthonormality of $\{x_l\}_l$ and $\{\tilde{x}_l\}_l$, we have for $s\in[K]$,
	$$
	\langle x_s,y_k\rangle
	=\sum_{l=1}^n(x_s)_l\left(\sum_{i=1}^K(x_i)_l O_{ki}\right)
	=\sum_{i=1}^K\left(\sum_{l=1}^n(x_s)_l(x_i)_l\right)O_{ki}
	=O_{ks},
	$$
	and for $t\in[\km]\setminus[K]$,
	$$
	\langle\tilde{x}_t,\tilde {y}_k\rangle
	=\sum_{l=1}^{n}(\tilde{x}_t)_l\left(\sum_{i=1}^K(\tilde{x}_i)_lO_{ik}\right)
	=\sum_{i=1}^K O_{ik}\left(\sum_{l=1}^{n}(\tilde{x}_t)_l(\tilde{x}_i)_l\right)
	=\sum_{i=1}^K O_{ik} \mathbf{1}_{\{i=t\}}=0.
	$$
\end{proof}

\paragraph{Bounds on population quantities}
The previous lemma implies bounds on the population
quantity of interest, $\lambda_P(\tilde{x}_l)$.
\begin{lemma}[Bounding $\lambda_P(\tilde{x}_l)$]
	\label{lem:dom_lb}
	Under 
	event $\mathcal E^0$, 
	\begin{align*}
		&{\tilde{x}}_l\T P{\tilde{x}}_l \gtrsim  n\rho_n,\,\,\,\qquad \forall l\in[K], \\
		&{\tilde{x}}_l\T P{\tilde{x}}_l\lesssim \log n,\qquad \forall l\in[\km]\setminus[K] .
	\end{align*}
\end{lemma}
\begin{proof}
	For $s\in[K]$, expanding ${\tilde{x}}_s$ over an orthonormal basis including $\{x_l\}_{l\in K}$, we get
	\begin{eqnarray}
		{\tilde{x}}_s\T P{\tilde{x}}_s
		&=&\sum_{k=1}^K\lambda_k\langle \tilde{x}_s, x_k\rangle^2\nonumber\\
		&=&\sum_{k=1}^K\lambda_k\left[\langle x_k,{y}_s\rangle^2-\langle x_k,y_s-\tilde{x}_s\rangle\langle x_k,\tilde{x}_s+y_s\rangle\right]\nonumber\\
		&\geq& \sum_{k=1}^K\lambda_k O_{sk}^2-\sum_{k=1}^K\lambda_k\Vert y_s-\tilde{x}_s\Vert_2\Vert x_k\Vert_2^2\left(\Vert\tilde{x}_s\Vert_2+\Vert y_s\Vert_2\right)\label{eqn:mu_s1}\\
		&\geq& \psi_1^{-1}\psi_1'n\rho_n-O\left(2 Kn\rho_n\sqrt{\frac{\log n}{n\rho_n}}\right)\label{eqn:mu_s2}\\
		&\gtrsim&n\rho_n\nonumber
	\end{eqnarray}
	where inequality \eqref{eqn:mu_s1} follows from Cauchy–Schwarz, the triangle inequality and $\langle x_k,{y}_s\rangle^2=O_{sk}$ by Lemma \ref{lem:eignvector_2}. Inequality \eqref{eqn:mu_s2} holds since $\sum_{k=1}^K O_{sk}^2=1$, $\psi_1^{-1}\psi_1'n\rho_n\leq\lambda_k\leq n\rho_n$ by Assumption~\ref{assump:matrix_SSNR},  $\Vert\tilde{x}_s-y_s\Vert_2=O\left(\sqrt{\frac{\log n}{n\rho_n}}\right)$ by Lemma \ref{lem:eignvector_2} and $\norm{\tilde x_k}_2=\norm{x_k}_2=\norm{y_s}_2=1$.
	
	For $t\in [\km]\setminus[K]$,
	\begin{eqnarray}
		\tilde{x}_t\T   P\tilde{x}_t
		&=&\sum_{k=1}^K\lambda_k\langle\tilde{x}_t, {x}_k\rangle^2\nonumber\\
		&=&\sum_{k=1}^K\lambda_k\langle\tilde{x}_t, {x}_k-\tilde{y}_k+\tilde {y}_k\rangle^2\nonumber\\
		&=& \sum_{k=1}^K\lambda_k\left[\langle\tilde{x}_t, {x}_k-\tilde{y}_k\rangle+\langle\tilde{x}_t,\tilde {y}_k\rangle\right]^2\nonumber\\
		&=&\sum_{k=1}^K\lambda_k\langle\tilde{x}_t, {x}_k-\tilde{y}_k\rangle^2 \label{eqn:mu_t_0}\\
		&\leq& K\lambda_1\max_{k\in[K]}\Vert  {x}_k-\tilde{y}_k\Vert_2^2=O(\log n) \label{eqn:mu_t}
	\end{eqnarray}
	where equality $\eqref{eqn:mu_t_0}$ follows from $\langle \tilde{x}_t, \tilde{y}_k\rangle=0$ by Lemma \ref{lem:eignvector_2}. Equation \eqref{eqn:mu_t} holds since $\norm{\tilde{ y}_k-{x}_k}_2=O\left(\sqrt{\log n/{n\rho_n}}\right)$ by Lemma \ref{lem:eignvector_2} and $\lambda_k\leq \lambda_1\leq n\rho_n$ by Assumption~\ref{assump:matrix_SSNR}.
\end{proof}

\paragraph{Delocalization of signal eigenvectors} To establish concentration
of the estimate $\tilde\lambda_{\test}({\tilde x}_l)$ around $\varepsilon\lambda_P({\tilde x}_l)$ for $l \in [K]$, we first need to show that ${\tilde x}_l$ is delocalized. That result essentially follows from an entrywise version of Lemma~\ref{lem:eignvector_2} based on a technical result of \cite{abbe2020entrywise}. 
\begin{lemma}[Delocalization of signal sample eigenvectors]
	\label{lem:eigen_p}
	There exist constants $\delta_1>0$, $C_1>0$ such that the event
	\begin{align*}
		{\mathcal E}^1=&\left\{\Vert \tilde{x}_l\Vert_\infty\leq C_1\sqrt{\frac{\mu_0}{n}}, \forall \,l\in[K]\right\}
	\end{align*}
	holds with probability at least $1-3n^{-\delta_1}$.
\end{lemma}
\begin{proof}
	We use \citet[Theorem 2.1]{abbe2020entrywise} on $\tilde A$, which requires four conditions. We check these conditions next.
	First, let ${\tilde A}^* = (1-\varepsilon) P$,
	$\Delta^* = \lambda_K$,  
	\begin{equation}
		\label{eq:kappa-condition}
		\kappa=\frac{\lambda_1}{\lambda_K}\leq \psi_1,    
	\end{equation}
	where the inequality follows from Assumption~\ref{assump:matrix_SSNR}, 
	$$
	\varphi(x)=\frac{1}{32\psi_1}\min\{\sqrt{n}x, 1\},
	$$
	and
	\begin{equation}
		\label{eq:gamma-condition}
		\gamma=C''' \psi_1(\psi_1')^{-1} \sqrt{\frac{\log n}{n\rho_n}}\gtrsim \sqrt{\frac{\log n}{n^{{1-\xi_1}}}},
	\end{equation}
	where $C'''$ is the constant in Lemma~\ref{lem:concentration_p} and $\psi_1, \psi_1' >0$, $\xi_1\in(0,1)$ are the constants in Assumption~\ref{assump:matrix_SSNR}. 
	
	\begin{itemize}
		\item[(A1)] \emph{(Incoherence)} 
		By \citet[Eq.~(2.4)]{abbe2020entrywise} and the remarks that follow it, the incoherence condition is satisfied provided
		\begin{align*}
			\mu(U) := \frac{n}{K}\Vert U\Vert_{2,\infty}^2 \leq \frac{n \gamma^2}{K \kappa^2}.
		\end{align*}
		Under Assumption~\ref{assump:coherence}, $\mu(U) \leq \mu_0$ while~\eqref{eq:gamma-condition} implies $n\gamma^2=\Omega(\log n)$ and~\eqref{eq:kappa-condition} implies $\kappa = O(1)$. Hence the condition is satisfied.
		
		\item[(A2)] \emph{(Row and columnwise independence)}
		By Lemma~\ref{lemma:poisson_binomial}, ${\tilde A}$ is the adjacency matrix of a Poisson graph with independent entries. In particular, $\{{\tilde A}_{ij} : \text{$i = m$ or $j = m$}\}$ are independent of $\{{\tilde A}_{ij} ; i \neq m, j \neq m\}$.
		
		\item[(A3)] \emph{(Spectral norm concentration)} 
		As observed previously, applying Lemma~\ref{lem:concentration_p} with $B := \tilde{A}$, $Q := (1-\varepsilon) P$ and $\delta > 0$ shows that the event
		\begin{align*}
			&\mathcal E^0=\left\{\norm{\tilde A-(1-\varepsilon)P}\leq C'''\sqrt{n\rho_n\log n}\right\},
		\end{align*}
		holds with probability $1 - n^{-\delta}$.
		Moreover, by the remark after Assumption~\ref{assump:matrix_SSNR}, 
		$$
		\gamma \Delta^*
		= C''' \psi_1(\psi_1')^{-1} \sqrt{\frac{\log n}{n\rho_n}} \lambda_K \geq C''' \sqrt{n\rho_n \log n}.
		$$
		Hence,
		$$
		\Prob\left[\norm{{\tilde A}-{\tilde A}^*}\leq \gamma \Delta^*\right]
		\geq 1 - n^{-\delta}.
		$$
		Note further that, under Assumption~\ref{assump:matrix_SSNR}, $\gamma = o(1)$, which implies
		$$
		32 \kappa \max\{\gamma, \varphi(\gamma)\}
		\leq 32 \kappa \max\left\{\gamma, \frac{1}{32 \psi_1}\right\}
		\leq 1,
		$$
		for $n$ large enough,
		as required in \citet[Assumption (A3)]{abbe2020entrywise}, where we used~\eqref{eq:kappa-condition}.
		
		\item[(A4)] \emph{(Row concentration)} As required in \citet[Assumption (A4)]{abbe2020entrywise}, the function $\varphi$ is continuous and non-decreasing on $\R_+$ with $\varphi(0) = 0$ and $\varphi(x)/x$ nonincreasing on $\R_+$. Let $W\in\mathbb R^{n\times K}$. By standard norm bounds
		\begin{align*}
			\frac{1}{\sqrt{n}}\leq\frac{\norm{W}_F}{\sqrt{n}\norm{W}_{2,\infty}}\leq 1.
		\end{align*}
		As a result, by definition of $\varphi$, 
		$$
		\varphi\left(\frac{\norm{W}_F}{\sqrt{n}\norm{W}_{2,\infty}}\right) = \frac{1}{32 \psi_1}.
		$$
		Let
		$$
		g = \Delta^* \norm{W}_{2,\infty} \varphi\left(\frac{\norm{W}_F}{\sqrt{n}\norm{W}_{2,\infty}}\right)
		= \frac{1}{32 \psi_1}\lambda_K\norm{W}_{2,\infty}.
		$$
		Fix $m\in[n]$ and $r\in[K]$. 
		Applying Lemma~\ref{lem:bernstein-poisson} on $\tilde A_{m\cdot}$ with $\max_{ij}\mathbb E\tilde A_{ij}\leq (1-\varepsilon)\rho_n$, there exist $c_2>0$, $c_2'>1$ such that
		\begin{eqnarray*}
			&&\mathbb P\left(\left\vert\sum_{i\in[n]}(\tilde A_{mi}-\tilde Q_{mi})W_{ir}\right\vert\geq {g}/{\sqrt{K}}\right)\\
			&\leq& 2\exp\left({-\frac{g^2/K}{C' (1-\varepsilon) \rho_n \|W_{\cdot r}\|_2^2+C''\norm{W_{\cdot r}}_{\infty}g/\sqrt{K}}}\right)\\
			&=& 2\exp\left({-\frac{\lambda_K^2 \norm{W}_{2,\infty}^2}{32^2\psi_1^2K C' (1-\varepsilon) \rho_n \|W_{\cdot r}\|_2^2+ 32\psi_1\sqrt{K} C''\norm{W_{\cdot r}}_{\infty} \lambda_K \norm{W}_{2,\infty} }}\right)\\
			&\leq& 2\exp\left({-\frac{\lambda_K^2 }{32^2\psi_1^2K C' (1-\varepsilon) n \rho_n + 32\psi_1\sqrt{K} C'' \lambda_K }}\right)\\
			&\leq& 2\exp(-c_2n\rho_n)\\
			&\leq& n^{-c_2'},
		\end{eqnarray*}
		where $C'$ and $C''$ are the constants in Lemma~\ref{lem:bernstein-poisson} and we used again that, by the remark after Assumption~\ref{assump:matrix_SSNR}, $\lambda_K \geq \psi_1^{-1}\psi_1'n\rho_n$.
		In the final inequality, we use that $n\rho_n \geq c_0\log^{\xi_0} n$, $\xi_0>2$ under Assumption~\ref{assump:matrix_SSNR}. Since 
		$$
		\norm{(\tilde A-\tilde Q)_{m\cdot}W}_2\leq \sqrt{K}\sup_r\left\vert\sum_{i\in[n]} ({\tilde A}_{mi}-{\tilde Q}_{mi}){W}_{ir}\right\vert,
		$$ 
		a union bound over $r$ implies
		$$
		\Prob\left[
		\norm{(\tilde A-\tilde Q)_{m\cdot}W}_2
		\leq g\right]
		\geq 1 - Kn ^{-c_2'}.
		$$
	\end{itemize}

	Recall the definition of $\hat U$ and $U$ from~\eqref{eq:def-uhat-u}. Applying 
	\citet[Theorem 2.1]{abbe2020entrywise} and using \citet[Eq.~(2.4)]{abbe2020entrywise} again, 
	there exists $\tilde C>0$ such that
	\begin{eqnarray*}
		\max_{l\in [K]}\norm{\tilde{x}_l}_\infty
		&\leq& \norm{\hat U}_{2,\infty}\\
		&\leq& \tilde C(2\kappa+\varphi(1))\norm{ U}_{2,\infty}\\
		&\leq& \tilde C\left(2 \psi_1 + \frac{1}{32\psi_1}\right)\sqrt{K}\sqrt{\frac{\mu_0}{n}},
	\end{eqnarray*}
	with probability $1 - n^{-\delta} - 2n^{-(c_2'-1)}$,
	where we used Assumption~\ref{assump:coherence} on the last line.
	Taking $C_1=\tilde C(2\psi_1 + \frac{1}{32\psi_1})\sqrt{K}$ and $\delta_1=\min\{\delta,c_2'-1\}>0$ gives the claim.
\end{proof}

\paragraph{Concentration of quadratic forms}
We bound the variance estimate for the signal eigenvectors.
\begin{lemma}[Bound on the variance estimate]
	\label{lem:sigma}
	Under event $\mathcal E^0\cap\mathcal E^1$, for all $l\in[K]$, 
	$$
	\tilde \sigma^2_l
	:=
	\frac{\varepsilon}{1 - \varepsilon}(\tilde{x}_l^2)\T\left(2\tilde A -\diag(\tilde A)\right)\tilde{x}_l^2 = \Theta(\rho_n).
	$$
\end{lemma}
\begin{proof}
	Let ${\tilde Q} = (1- \varepsilon) P$.
	We first show $(\tilde{x}_l^2)\T {\tilde A}\tilde{x}_l^2$ can be controlled via $(\tilde{x}_l^2)\T {\tilde Q}\tilde{x}_l^2$. Indeed
	observe that for each $l\in[K]$
	\begin{eqnarray}
		\left\vert(\tilde{x}_l^2)\T {\tilde A}\tilde{x}_l^2-(\tilde{x}_l^2)\T {\tilde Q}\tilde{x}_l^2\right\vert
		&=&\vert(\tilde{x}_l^2)\T( {\tilde A}-{\tilde Q})\tilde{x}_l^2\vert\nonumber\\ 
		&\leq&\norm{{\tilde A}-{\tilde Q}}\norm{\tilde{x}_l^2}^2_2\nonumber\\
		&\leq&\norm{{\tilde A}-{\tilde Q}}  \Vert\tilde{x}_l\Vert_\infty^2\Vert\tilde{x}_l\Vert_2^2\nonumber\\
		&=&O\left(\sqrt{n \rho_n \log n} \cdot \frac{1}{n}\right)\nonumber\\
		&=&O\left(\sqrt{\frac{\rho_n \log n}{n}}\right)\nonumber
	\end{eqnarray}
	where we used that
	$\norm{{\tilde A}-{\tilde Q}}=O(\sqrt{n\rho_n \log n})$ under event $\mathcal E^0$ and $\norm{\tilde{x}_l^2}_\infty=\norm{\tilde{x}_l}_\infty^2=O(\frac{1}{n})$ under $\mathcal E^1$. Moreover, observe that $\sqrt{\rho_n \log n/n} \ll \rho_n$ since $n\rho_n\geq c_0\log^{\xi_0}n$ under Assumption~\ref{assump:matrix_SSNR}. So
	\begin{eqnarray}
		\left\vert(\tilde{x}_l^2)\T {\tilde A}\tilde{x}_l^2-(\tilde{x}_l^2)\T {\tilde Q}\tilde{x}_l^2\right\vert
		&\ll&\rho_n. \label{eqn:sigma_ub_2}
	\end{eqnarray}
	
	To get an upper bound on $\tilde \sigma^2_l$, note that
	\begin{eqnarray}
		(\tilde{x}_l^2)\T P\tilde{x}_l^2&\leq& \lambda_1\Vert\tilde{x}_l^2\Vert_2^2\nonumber\\ &\leq& \lambda_1\cdot\Vert\tilde{x}_l\Vert_\infty^2\cdot\Vert\tilde{x}_l\Vert_2^2\nonumber\\
		&=&O\left(n\rho_n\cdot\frac{1}{n}\cdot 1\right)\nonumber\\
		&=&O(\rho_n),\nonumber
	\end{eqnarray}
	where we used $\lambda_1\leq n\rho_n$ by Assumption~\ref{assump:matrix_SSNR}.
	Hence, we get
	\begin{eqnarray*}  
		\tilde\sigma^2_l
		&=&\frac{\varepsilon}{1- \varepsilon} \left[2(\tilde{x}_l^2)\T\tilde A \tilde{x}_l^2-(\tilde{x}_l^2)\T\diag(\tilde A)\tilde{x}_l^2\right]\\
		&\leq&\frac{2\varepsilon}{1- \varepsilon} (\tilde{x}_l^2)\T\tilde A \tilde{x}_l^2\\
		&\leq&\frac{2\varepsilon}{1- \varepsilon}\vert (\tilde{x}_l^2)\T {\tilde A}\tilde{x}_l^2-(\tilde{x}_l^2)\T {\tilde Q}\tilde{x}_l^2\vert
		+\frac{2\varepsilon}{1- \varepsilon}(\tilde{x}_l^2)\T {\tilde Q}\tilde{x}_l^2\\
		&\leq&\frac{2\varepsilon}{1- \varepsilon}\vert (\tilde{x}_l^2)\T {\tilde A}\tilde{x}_l^2-(\tilde{x}_l^2)\T {\tilde Q}\tilde{x}_l^2\vert
		+ 2\varepsilon (\tilde{x}_l^2)\T P\tilde{x}_l^2\\
		&=& O(\rho_n),
	\end{eqnarray*}
	by~\eqref{eqn:sigma_ub_2}.
	
	In the other direction, 
	by Cauchy-Schwarz, 
	\begin{align*}
		&(\tilde{x}_l^2)\T P\tilde{x}_l^2
		\geq \frac{\left(\tilde{x}_l\T P\tilde{x}_l\right)^2}{\sum_{ij}P_{ij}}
		\gtrsim \frac{(n \rho_n)^2}{n^2\rho_n}
		\gtrsim \rho_n,
	\end{align*}
	where the middle inequality follows from Lemma~\ref{lem:dom_lb}.
	Combining with \eqref{eqn:sigma_ub_2}, we have 
	\begin{eqnarray*}
		\tilde\sigma^2_l
		&=&\frac{\varepsilon}{1- \varepsilon} \left[(\tilde{x}_l^2)\T\tilde A \tilde{x}_l^2+(\tilde{x}_l^2)\T\left(\tilde A -\diag(\tilde A)\right)\tilde{x}_l^2\right]\\
		&\geq&\frac{\varepsilon}{1-\varepsilon} (\tilde{x}_l^2)\T {\tilde A}\tilde{x}_l^2\\
		&\geq&\frac{\varepsilon}{1-\varepsilon}(\tilde{x}_l^2)\T {\tilde Q}\tilde{x}_l^2
		- \frac{\varepsilon}{1-\varepsilon} \left\vert (\tilde{x}_l^2)\T {\tilde A}\tilde{x}_l^2-(\tilde{x}_l^2)\T {\tilde Q}\tilde{x}_l^2\right\vert\\
		&=& \varepsilon(\tilde{x}_l^2)\T P \tilde{x}_l^2
		- \frac{\varepsilon}{1-\varepsilon} \left\vert (\tilde{x}_l^2)\T {\tilde A}\tilde{x}_l^2-(\tilde{x}_l^2)\T {\tilde Q}\tilde{x}_l^2\right\vert\\ 
		&\gtrsim& \rho_n.
	\end{eqnarray*}
	That concludes the proof.
\end{proof}

Finally, before moving on to the proof of the main theorem, we show next that $\tilde\lambda_{\test}({\tilde x}_l)$ is concentrated around $\varepsilon\lambda_P({\tilde x}_l)$.
\begin{lemma}[Concentration of $\tilde\lambda_{\test}(x)$]
	\label{lem:concen_A_test_p}
	Let $x \in \R^n$ be a unit vector such that
	\begin{equation}
		\label{eq:lambdatest-deloc}
		\norm{x}^2_\infty
		\leq
		\frac{\log n}{n},
	\end{equation}
	then there exists $\delta_2>1$ such that
	\begin{align*}
		\Prob\left[\left|\sum_{i,j} x_i x_j({\tilde A}\test - \varepsilon P)_{ij}\right|\leq \sqrt{\rho_n \log n} \right] \geq 1 - n^{-\delta_2}.
	\end{align*}
\end{lemma}
\begin{proof}
	We use Lemma~\ref{lem:bernstein-poisson}.
	From $\|x\|_2 = 1$, we get
	\begin{align*}
		&\Prob\left[\left|\sum_{i,j} x_i x_j({\tilde A}\test - \varepsilon P)_{ij}\right|
		\geq \sqrt{\rho_n \log n} \right] \\
		&\leq 2 \exp\left(
		- \frac{(\sqrt{\rho_n \log n})^2/2}{C' \varepsilon \rho_n \sum_{i,j} (x_i x_j)^2 + C'' \max_{ij} |x_i x_j| \, \sqrt{\rho_n \log n}}
		\right)\\
		&\leq 2 \exp\left(
		- \frac{\rho_n \log n/2}{C' \varepsilon \rho_n + C'' \|x\|_\infty^2 \, \sqrt{\rho_n \log n}}
		\right).
	\end{align*}
	By Assumption~\ref{assump:matrix_SSNR},
	$\rho_n \gg \frac{\log n}{n}$ while 
	$\sqrt{\rho_n \log n} = o(1)$. By~\eqref{eq:lambdatest-deloc}, the denominator on the last line is $\lesssim \rho_n$ and the claim follows.
\end{proof}

\subsubsection{Proof of Theorem~\ref{thm:consistency}}
\label{sec:proof_poisson}

Now, we are ready to prove Theorem \ref{thm:consistency}. 
\begin{proof}[Proof of Theorem~\ref{thm:consistency}]
	By Lemmas~\ref{lem:concentration_p} and~\ref{lem:eigen_p},
	the event $\mathcal E^0 \cap \mathcal E^1$ holds with
	probability $1 - 4 n^{-\delta_1}$. Under $\mathcal E^0 \cap \mathcal E^1$, which depends only on ${\tilde A}$, the claims
	in Lemmas~\ref{lem:eignvector_2},~\ref{lem:dom_lb} and~\ref{lem:sigma} also hold. For the rest of the proof, 
	we condition on $\mathcal E^0 \cap \mathcal E^1$ and use the fact that ${\tilde A}\test$ is independent of ${\tilde A}$ by Lemma~\ref{lemma:poisson_binomial}.

	Let ${\tilde x}_l$, $l \in [\km]$, be the
	top $\km$ unit eigenvectors of ${\tilde A}$
	and let
	$$
	\tilde\sigma^2_l
	=\frac{\varepsilon}{1- \varepsilon} (\tilde{x}_l^2)\T \left(2{\tilde A}-\diag(\tilde A)\right)\tilde{x}_l^2.
	$$
	Define
	$$
	S
	=\left\{l\in[\km]:\norm{\tilde{x}_l}^2_\infty\leq\min\left\{\frac{\tilde\sigma^2_l}{\log^2 n}, \frac{\log n}{n}\right\}\right\},
	$$
	to be the subset of $[\km]$ corresponding to sufficiently delocalized eigenvectors. Recall that the test statistic associated to ${\tilde x}_l$ is 
	$$
	T_l
	=\frac{\tilde{x}_l\T \tilde A\test\tilde{x}_l}{\tilde \sigma_l}.
	$$
	We say that $l$ is rejected if  
	$$
	l \in S
	\quad\text{and}\quad
	|T_l| \geq \sqrt{n \log n} =: \tau_n. 
	$$
	
	\paragraph{No under-estimation}
	We show that the test statistic associated with the $K$ leading eigenvectors of $\tilde A$ will reject the null hypothesis with high probability, that is,
	\begin{itemize}
		\item $[K]\subset S$; and
		\item $|T_l|\geq \tau_n$, $\forall l\in[K]$.
	\end{itemize}
	
	Fix $s\in[K]$. First, we check that $s\in S$. Under ${\mathcal E^1}$, $\norm{\tilde x_s^2}_\infty=O(1/n) \ll \log n/n$. We need to check that $\norm{\tilde{\underline{x}}_{s}^2}_{\infty}\leq \tilde\sigma^2_s/\log^2 n$, for $n$ sufficiently large. This follows from the fact that $\tilde\sigma^2_s=\Theta(\rho_n)$ by Lemma~\ref{lem:sigma} and $\rho_n\geq c_0 n^{-1} \log^{\xi_0}n$ with $\xi_0 > 2$ under Assumption~\ref{assump:matrix_SSNR}.  
	
	Next, we bound $|T_s|$ from below. We have, with probability $1 - n^{-\delta_2}$, where $\delta_2>1$ is the constant in Lemma~\ref{lem:concen_A_test_p},
	\begin{eqnarray}
		\vert T_s\vert
		&=& \left\vert\frac{\tilde{x}_s\T \tilde A\test\tilde{x}_s}{\tilde \sigma_s}\right\vert\nonumber\\
		&\geq& \frac{\varepsilon\left\vert\tilde{x}_s\T P\tilde{x}_s\right\vert-\left\vert{\tilde{x}_s\T (\tilde{A}\test- \varepsilon P)\tilde{x}_s}\right\vert}{\tilde\sigma_s}\nonumber\\
		&\gtrsim& {\frac{{\varepsilon}n\rho_n-\sqrt{\rho_n\log n}}{\sqrt{\rho_n}}}\nonumber\\
		&\gtrsim& n\sqrt{\rho_n}\nonumber\\
		&\gg& \sqrt{n \log n}, \label{eqn:T_lb}
	\end{eqnarray}
	where the dominating term is controlled through $\vert\tilde{x}_s\T P\tilde{x}_s\vert\gtrsim n\rho_n \gg \log n$ 
	by Lemma~\ref{lem:dom_lb}, the term $\vert{\tilde{x}_s\T (\tilde{A}\test- \varepsilon P)\tilde{x}_s}\vert$ is bounded above by $\sqrt{\rho_n\log n} \ll \log n$ from Lemma~\ref{lem:concen_A_test_p} and the denominator satisfies $\tilde\sigma^2_s=\Theta(\rho_n)$ by Lemma~\ref{lem:sigma}. The final bound follows from Assumption~\ref{assump:matrix_SSNR}.
	By a union bound, \eqref{eqn:T_lb} 
	holds simultaneously for $s \in [K]$ 
	with probability $1 - K n^{-\delta_2}$.
	
	\paragraph{No over-estimation}
	Then, we show that the noise eigenvectors of $\tilde A$ will either be too localized or the test statistic associated with them will fail to reject the null hypothesis. In other words, we show that for any $s\in S\setminus [K]$, it holds that $|T_s|<\tau_n$ with high probability. 
	
	Let $t\in S\setminus[K]$. We bound $|T_t|$ from above as follows
	\begin{eqnarray}
		\vert T_t\vert&=&\left\vert\frac{\tilde{x}_t\T \tilde A\test\tilde{x}_t}{\tilde \sigma_t}\right\vert\nonumber\\
		&\leq& \frac{\varepsilon\left\vert\tilde{x}_t\T P\tilde{x}_t\right\vert+\left\vert{\tilde{x}_t\T (\tilde{A}\test- \varepsilon P)\tilde{x}_t}\right\vert}{\tilde\sigma_t}\label{eqn:Tub}\\
		&=&O\left(\sqrt{\frac{n}{\log^2 n}}\cdot(\log n+\sqrt{\rho_n\log n})\right)\nonumber\\
		&=&O\left(\sqrt{{n}}\right). \label{eqn:T_ub}
	\end{eqnarray}
	The first term in the numerator of \eqref{eqn:Tub} satisfies $\left\vert{{x}_t\T P{x}_t}\right\vert=O(\log n)$ by Lemma~\ref{lem:dom_lb} while the term 
	$\vert{\tilde{x}_t\T (\tilde{A}\test- \varepsilon P)\tilde{x}_t}\vert$ 
	in \eqref{eqn:Tub} is bounded above by $\sqrt{\rho_n\log n} \ll \log n$ from Lemma~\ref{lem:concen_A_test_p}. For the denominator $\tilde\sigma_t$, $t\in S$ implies that $\Vert \tilde x_t^2\Vert_\infty\leq \tilde \sigma^2_t/\log^2 n$, thus
	\begin{align*}
		\tilde \sigma^2_t\geq \log^2 n\cdot\Vert \tilde x_t^2\Vert_\infty\geq \frac{\log^2 n}{n}\cdot n\Vert \tilde x_t^2\Vert_\infty&\geq\frac{\log^2 n}{n}\cdot\Vert \tilde x_t\Vert_2^2=\frac{\log^2 n}{n}.
	\end{align*}
	By a union bound, \eqref{eqn:T_ub} holds simultaneously for $t \in S\setminus[K]$ 
	with probability at least $1 -  n^{-\delta_2+1}$.
	
	\paragraph{Consistency}
	Therefore, it follows that the algorithm outputs $\hat K=K$ with probability tending to 1.
\end{proof}

\begin{remark}\label{remark:consistency_bernoulli}
    A similar result can be obtained in the Bernoulli case. As we explained in the context of the CLT, 
    one loses the independence
    of ${\tilde A}\test$ and ${\tilde A}$ which is used to apply Lemma~\ref{lem:concen_A_test_p} in the proof of the theorem.
    But one can instead condition on which pairs of vertices $\Omega \subseteq \{(i,j) : 1 \leq i < j \leq n\}$ are set aside for the test graph -- before generating $\tilde A$ and $\tilde A\test$. In that case, the two graphs are conditionally independent. The key difference is that 
    $\varepsilon P$ is replaced with the conditional version $\mathcal{P}_\Omega P$ in Lemma~\ref{lem:concen_A_test_p}, where recall that $\mathcal{P}_\Omega P \in \mathbb{R}^{n \times n}$ is the matrix with entries 
$$
(\mathcal{P}_\Omega P)_{ij}
:= P_{ij} \mathbf{1}\{\text{$(i,j) \in \Omega$ or 
$(j,i) \in \Omega$}\}.
$$
One can then use Lemma 2 in the supplementary materials of~\citep{li_network_2020} to show that the two mean quantities are close.
We omit the details. 
\end{remark}

\section{Details on motivating examples} \label{sec:motivation_details}

In the introduction, the simulated graph comes from a Degree-Corrected Stochastic Blockmodel (DCSBM).  See Section \ref{sec:numerical_simulations} for a description of this model and its parameters.  In Figures \ref{fig:motivation_scree} and \ref{fig:motivation_simulation}, the DCSBM has $k=128$ blocks.  
The 2,560 nodes are randomly assigned to the 128 blocks with uniform probabilities. On average, each block contains 20 nodes. The smallest block has 10 nodes and the largest block has 32. The degree parameters $\theta_i$ are distributed as $\text{Exponential}(\lambda =1)$ and the $B$ matrix is hierarchically structured. In order to specify the elements of $B$, let $\mathbb{T}$ be a complete binary tree with 7 generations (i.e., $2^7 = 128$ leaves).  Each leaf node is assigned to one of the $k=128$ blocks.  Define $u \wedge v \in \mathbb{T}$ as the most recent common ancestor of $u$ and $v$ (i.e., the node closest to the root along the shortest path between $u$ and $v$).  Define $g(u,v)$ as the distance in $\mathbb{T}$ from the root to $u \wedge v$.  
So, if the shortest path between $u$ and $v$ passes through the root, then $g(u,v) = 0$. Moreover, $g(u,u) = 7$ for all leaf nodes $u$.   Set $B_{u,v} = p 2^{g(u,v)}$, where $p = .0008$ is chosen so that the average expected degree of the nodes is equal to 20.

The citation graph was originally constructed and studied in \cite{rohe2020vintage}. There are roughly 220 million academic papers in the dataset.  Citations from one paper to another were converted to citations between the journals that published the papers.  If there were more than 5 citations from journal $i$ to journal $j$ using a $5\%$ sample of all edges, then $A_{ij}$ is set to one.  Otherwise, $A_{ij}$ is zero.  For simplicity, the graph was symmetrized by setting $A_{ij}=1$ if $A_{ji} = 1$.  
In this example, $\varepsilon = .05$ and for illustration, the data was divided only one time.

The illustration in Figures \ref{fig:motivation_scree} and \ref{fig:motivation_simulation} splits the edges ten separate times, each with probability $\varepsilon = .1$, and averages the results over those ten folds.

\section{Poisson vs. Bernoulli} \label{sec:poisson_vs_bernoulli}

The Poisson model has previously been used to study statistical inference with random graphs \citep{karrer_stochastic_2011, flynn2020profile,crane2018edge, cai2016edge, amini_biclustering}). In the sparse graph setting, these models produce very similar graphs (e.g. Theorem 7 in \cite{fastrg}).  Moreover, under the Poisson model, the edges are exchangeable \citep{cai2016edge, crane2018edge, fastrg}.  

This section shows that \texttt{EigCV} continues to perform well in settings where the elements of $A$ are Bernoulli random variables.
The technical results and derivations in this paper do not directly apply to this setting.  The key difficulty comes from the edge splitting. Under the Bernoulli model, $A_{ij} \in \{0,1\}$.  So, it is not possible for both $\tilde A$ and $\tilde A\test$ to get the edge $i,j$. This creates \textit{negative} dependence. 
Because the fitting and test graph are dependent,
a central limit theorem akin to Theorem \ref{thm:cltA} becomes more difficult (Theorem \ref{thm:bernoulli-clt}).  Nevertheless, a theorem akin to Theorem \ref{thm:consistency} still holds with minor changes to the conditions.

This simulation shows that the negative dependence created from Bernoulli edges makes the testing procedure more conservative.  That is, when we are testing 
\[H_0: \E(\lambda\test(\tilde x_\ell)) = 0 \]
for $\ell > k$, the $\alpha =.05$ test rejects with probability less than $.05$.  
This type of miscalibration is traditionally considered acceptable.

This section simulates from a $k=2$ Stochastic Blockmodel and examines the distribution of $T_3$ and $T_4$, under both the Bernoulli and Poisson models for edges.   In all of these simulations, there are $n=2000$ nodes.  Each node is randomly assigned to either block 1 or 2 with equal probabilities. 
Let $i$ and $j$ be any two nodes in the same block and $u$ and $v$ be any two nodes in different blocks.  Across simulation settings, $P_{ij}/P_{u,v} = 2.5$.
While keeping this ratio constant, the values in $P$ increase to make the expected degree of the graph go from 3.5 to 105, by increments of 3.5.  In the Poisson model, $A_{ij} \sim \text{Poisson}( P_{ij} )$ and in the Bernoulli model, $A_{ij} \sim \text{Bernoulli}( P_{ij} )$.
We compute $T_3$ and $T_4$ in \texttt{EigCV} (Algorithm \ref{alg:eigcv}) with edge split probability $\varepsilon = .05$ and $\mathrm{folds} = 1$. Supplementary Section \ref{app:bernoullivspoisson10fold} gives the identical simulation, but for $\mathrm{folds} = 10$.

We use the one-sided rejection region $T_\ell > 1.65$. If $T_\ell\sim N(0,1)$, then this has level $\alpha = .05$. We refer to the simulated probability that $T_\ell >1.65$ as the rejection probability.  
Figure \ref{fig:level} estimates the rejection probability in two ways.  First, each dot gives the proportion of 1000 replicates in which $T_\ell > 1.65$.  Second, the line gives the values 
\begin{equation} \label{eq:alphahat}
	\hat \alpha = 1 -\Phi((1.65-\bar T_\ell)/\mathrm{SD}(T_\ell)),
\end{equation}
where $\Phi$ is the cumulative distribution function (CDF) of the standard normal, $\bar T_\ell$ is the average value of $T_\ell$ over the 1000 replicates, and $\mathrm{SD}(T_\ell)$ is the standard deviation of these 1000 replicates.  
This is an estimate of the rejection probability, under the assumption that $T_\ell$ is normally distributed.

\begin{figure}[htbp] 
	\centering
	
	\textbf{Across simulation settings, the $\alpha = .05$ test is conservative.}
	
	\includegraphics[width=6in]{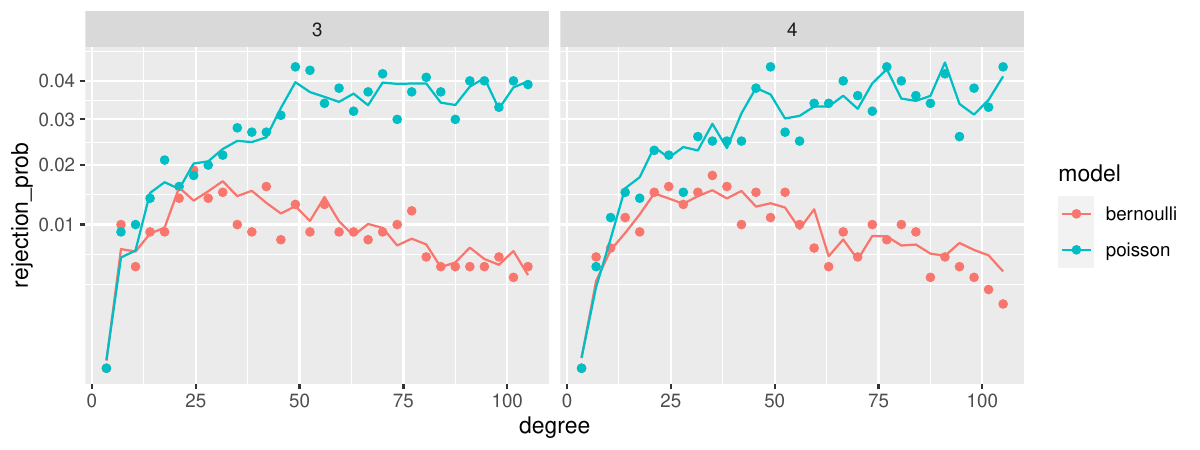} 
	\caption{The left panel gives results for $T_3$ and the right panel gives results for $T_4$. Each point corresponds to the proportion of 1000 replicates that the test statistic was greater than 1.65.  The line is a smoothed version of $\hat \alpha$. 
	}
	\label{fig:level}
\end{figure}

Under the simulation settings described above, Figure \ref{fig:level} shows four things. 
\begin{enumerate}
	\item The points and lines are  below .05.  So, the proposed test is conservative. 
	\item On the left side of both panels, 
	the points and lines are well below .05.  So, the proposed test is particularly conservative for very sparse graphs. Note that the bottom left point in both figures is over-plotted with two points. This is because none of the 4000 tests in the lowest degree graphs were rejected. The corresponding $\hat \alpha$'s are on the order 1/20,000.
	\item On the right side of both panels, 
	the red line decreases.  So,
	for Bernoulli graphs, the test is increasingly conservative for denser graphs. 
	\item The points are scattered around their respective lines, which suggests that the normal distribution provides a reasonable approximation for the right tail of the distribution (but $T_\ell$ does not have expectation zero and variance one). 
\end{enumerate}

The test is increasingly conservative for dense Bernoulli graphs.  In particular, the line giving the normal approximation $\hat \alpha$ is decreasing. So, either $\E(T_3)<0$ or $Var(T_3)<1$ or both (and similarly for $T_4$).  The next figure, Figure \ref{fig:average_z}, shows that the expectation of $T_3$ and $T_4$ decreases away from zero for dense Bernoulli graphs.

\begin{figure}[htbp] 
	\centering
	\textbf{Under the Bernoulli model, for $\ell>k$, the negative dependence makes the expectation of $T_\ell$ decrease away from zero as the graph becomes more dense. }
	
	\includegraphics[width=6in]{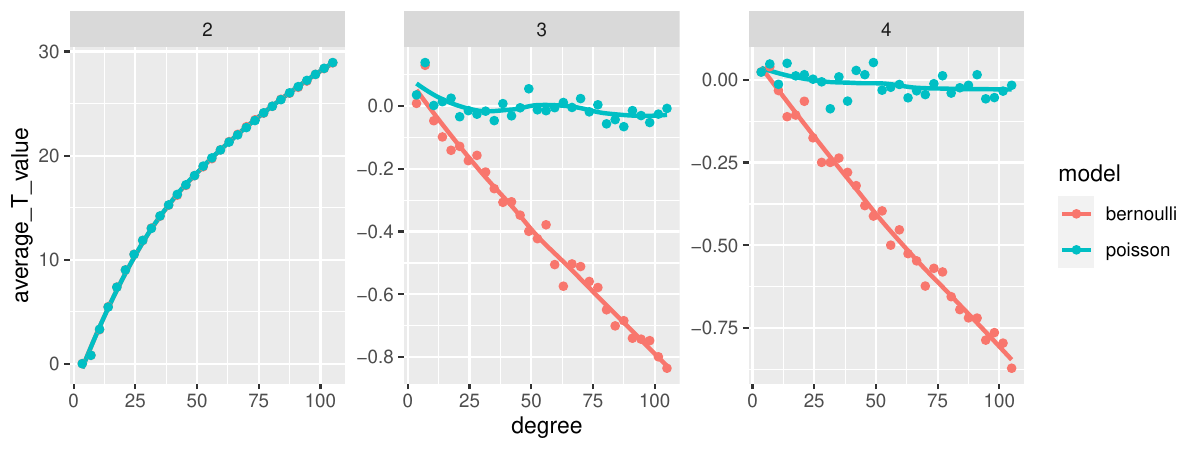} 
	\caption{ Each dot gives the average value of $T_\ell$ over 1000 replicates; $T_2$ on the left, $T_3$ in the middle, and $T_4$ on the right.  
		The previous figure, Figure \ref{fig:level}, shows that as the Bernoulli graph becomes denser, $T_3$ and $T_4$ are less likely to exceed 1.65. This figure shows that this is because their expectation decreases away from zero.
	}
	\label{fig:average_z}
\end{figure}

The expected values of $T_3$ and $T_4$ decrease due to the negative dependence between $\tilde A$ and $\tilde A\test$ under the Bernoulli model, as anticipated by the conditional nature of the CLT in Theorem \ref{thm:bernoulli-clt}. That is, if an edge appears in $\tilde A$, then it cannot appear in $\tilde A\test$, and vice versa. 
Figure \ref{fig:average_z} shows that $\tilde x_\ell\T \tilde A \tilde x_\ell$ and $\tilde x_\ell\T \tilde A\test \tilde x_\ell$ are  negatively correlated when (1) the edges are Bernoulli, (2) the graph is dense, (3) and $\ell>k$. 
In the two-block Stochastic Blockmodel, this negative dependence does not shift the expectation of $T_2$, even for dense Bernoulli graphs (see also Figure \ref{fig:difference} below). 
The fact that $\tilde x_\ell$ displays negative correlation, but only for $\ell>k$, is particularly interesting. Here is one interpretation: because these eigenvectors do not estimate signal, they only find noise in $\tilde A$.  Then, that noise disappears in $\tilde A\test$ due to the negative dependence. 

On the very left of both panels in Figure \ref{fig:level}, the expected degree is 3.5.  In this very sparse regime, which is below the weak recovery threshold \citep{mossel2015reconstruction} (see also \cite{abbe_community_2018}),  the test are conservative for both Bernoulli and Poisson graphs. 
In this very sparse regime, the eigenvectors $\tilde x_\ell$ often localize on a few nodes that have a large degree simply due to chance.  
When this happens,  $(\tilde x_\ell^2)\T {A} \tilde x_\ell^2$ over-estimates the population quantity $(\tilde x_\ell^2)\T {\E(A)} \tilde x_\ell^2$.  This causes $\tilde \sigma_\ell$ in the denominator of $T_\ell$ to be large, thus making the standard deviation of $T_\ell$ small and the test statistic less likely to exceed 1.65. Figure \ref{fig:sdsim} below displays the standard deviation of the test statistics. 
It is comforting that the test is  conservative in such scenarios because these localized eigenvectors $\tilde x$ are often particularly troubling artifacts of noise.  For example, the consistency result in Section \ref{sec:consistency} requires an additional step to \texttt{EigCV} that discards localized eigenvectors.  The current simulation suggests that the additional step is a technical requirement.  For this reason, we do not include the additional step in \texttt{EigCV}.

Supplementary Section \ref{app:bernoullivspoisson10fold} repeats these figures with  $\mathrm{folds} = 10$.  Increasing the number of folds makes the $T_3$ and $T_4$ even more conservative, while making $T_2$ more powerful.  This happens because the variation in $T_\ell$ comes from both the original data $A$ and the edge splitting. By increasing the number of folds, the second source of variation diminishes while not disrupting the expectation of $T_\ell$.

\subsection{Two clarifying figures for $\textrm{folds} = 1$}
The average value of $T_3$ and $T_4$ are different under the Poisson and Bernoulli models (see Figure \ref{fig:average_z} in Section \ref{sec:poisson_vs_bernoulli}).   To further illustrate this difference and to show that $T_2$ does not have the same property, 
Figure \ref{fig:difference} subtracts the average value of the test statistic \textit{under the Bernoulli model} from the corresponding average under the Poisson model. The difference appears for $T_3$ and $T_4$, but not for  $T_2$.
This shows that $\tilde x_\ell\T \tilde A \tilde x_\ell$ and $\tilde x_\ell\T \tilde A\test \tilde x_\ell$ are  negatively correlated when (1) the edges are Bernoulli, (2) the graph is dense, (3) and $\ell>k$. In the two-block Stochastic Blockmodel, this negative dependence does not shift the expectation of $T_2$, even for dense Bernoulli graphs.

\begin{figure}[htbp] 
	\centering
	\textbf{}
	
	\includegraphics[width=6in]{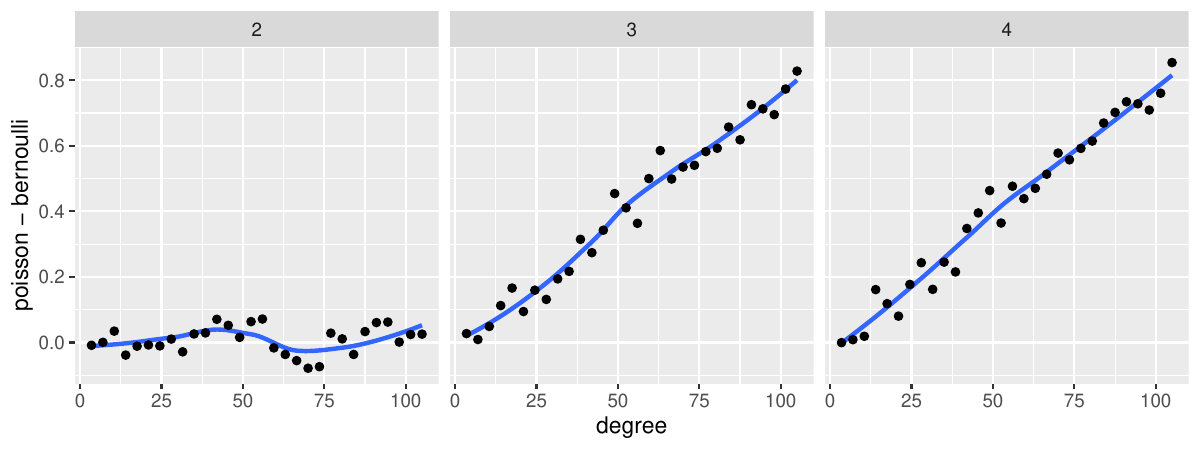} 
	\caption{ Each panel gives the difference between the Poisson and Bernoulli results from Figure \ref{fig:average_z}.   In this simulation model, there are $k=2$ blocks in the Stochastic Blockmodel.  Taken together, this suggests that the negative dependence in Bernoulli graphs between the fitting and testing adjacency matrices diminishes the expected value of $T_\ell$ when $\ell>k$ and the graph is more dense. However, the negative dependence does not appear to diminish the expected value of $T_2$. When a test is conservative under the null, one typically suffers a reduced power under the alternative.  However, this result suggests that the negative dependence require us to pay this price.  The Bernoulli model makes $T_3$ conservative, without making $T_2$ less powerful.
	}
	\label{fig:difference} 
\end{figure}

Figure \ref{fig:sdsim} displays the standard deviation for the test statistics $T_2, T_3,T_4$ over the 1000 replicates in the simulation.  This figure is repeated with $\mathrm{folds} = 10$ in Figure \ref{fig:sdsim10fold}.

\begin{figure}[htbp] 
	\centering
	\textbf{$SD(T_\ell)$ with $\mathrm{folds} = 1$}
	
	\includegraphics[width=6in]{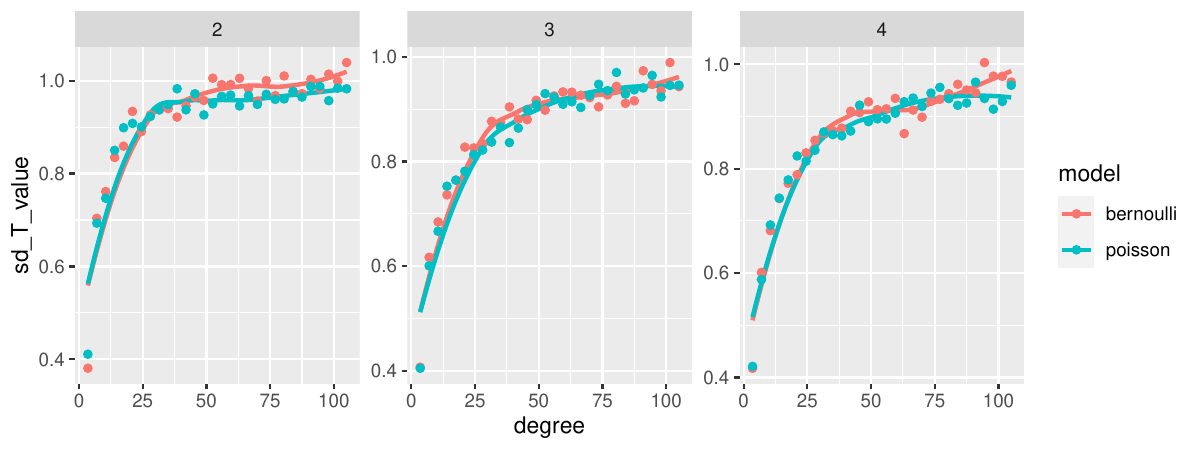} 
	\caption{ As discussed in Section \ref{sec:poisson_vs_bernoulli}, the localization of the eigenvectors on small degree graphs (on the left of each panel) makes $\tilde \sigma$ over estimate the standard error.  This makes the standard deviation of the test statistics over the 1000 replicates smaller than 1.  
	}
	\label{fig:sdsim} 
\end{figure}

\subsection{Increasing the number of folds in the Bernoulli vs. Poisson comparison} \label{app:bernoullivspoisson10fold}

Figure \ref{fig:rejection_prob_10fold} repeats Figure \ref{fig:level}, but with $\mathrm{folds}=10$ instead of 1.  It shows that increasing the number of folds makes $T_3$ and $T_4$ more conservative.  While Theorem \ref{thm:cltA} shows that $T_\ell$ is asymptotically normal with $\textrm{folds} = 1$, increasing the number of folds induces unknown dependence between the test statistics.  Figure \ref{fig:rejection_prob_10fold} suggests that even with 10 folds, the right tail of the distribution of $T_3$ and $T_4$ are well approximated by the normal distribution; this is because the bumpy line that gives $\hat \alpha$ is close to the points.  The tests are conservative because their expectation is not zero and their variance is not one.

Figure \ref{fig:sdsim10fold} shows that 
increasing the number of folds dramatically reduces the variation in $T_2, T_3, T_4$.  This makes $T_2$ more powerful and $T_3, T_4$ more conservative.

\begin{figure}[htbp] 
	\centering
	\textbf{}
	
	\includegraphics[width=6in]{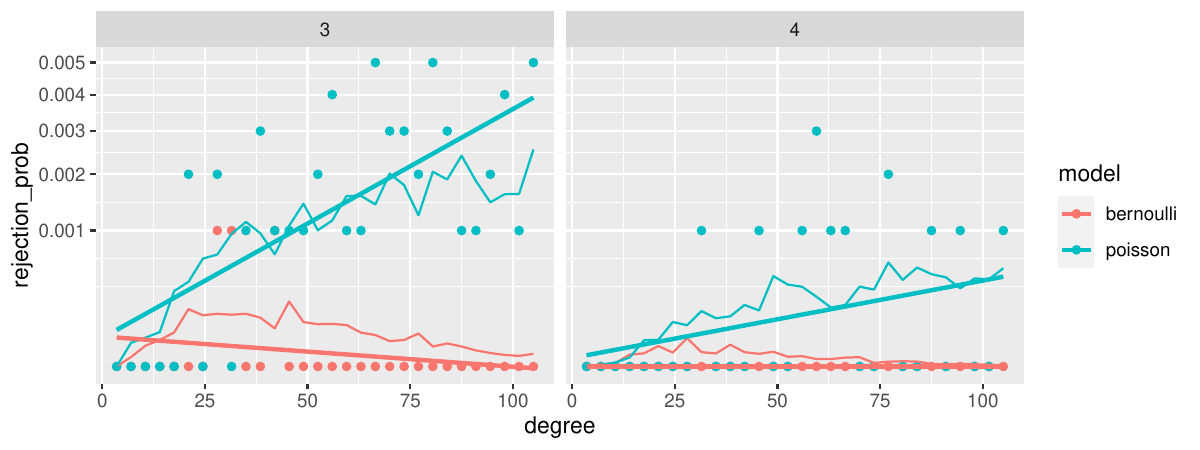} 
	\caption{ This is a repeat of Figure \ref{fig:level}, but with $\textrm{folds} = 10$ instead of 1.  The vertical axis gives estimates of the rejection probability for $T_3$ (left) and $T_4$ (right).  Each dot give the proportion of 1000 simulations for which the test statistic exceeds 1.65.  The bumpy line interpolates the values $\hat \alpha$, defined in Equation \eqref{eq:alphahat}.  Because there are so few rejections, particularly for $T_4$, it is difficult to see whether the bumpy line is close to the points.  The straight line is the ordinary least squares fit to the points. It roughly aligns with the bumpy line.  This suggests that the right tails of the distributions for $T_3$ and $T_4$ are approximately normally distributed. 
	}
	\label{fig:rejection_prob_10fold} 
\end{figure}

\begin{figure}[htbp] 
	\centering
	\textbf{$SD(T_\ell)$ with $\mathrm{folds = 10}$}
	
	\includegraphics[width=6in]{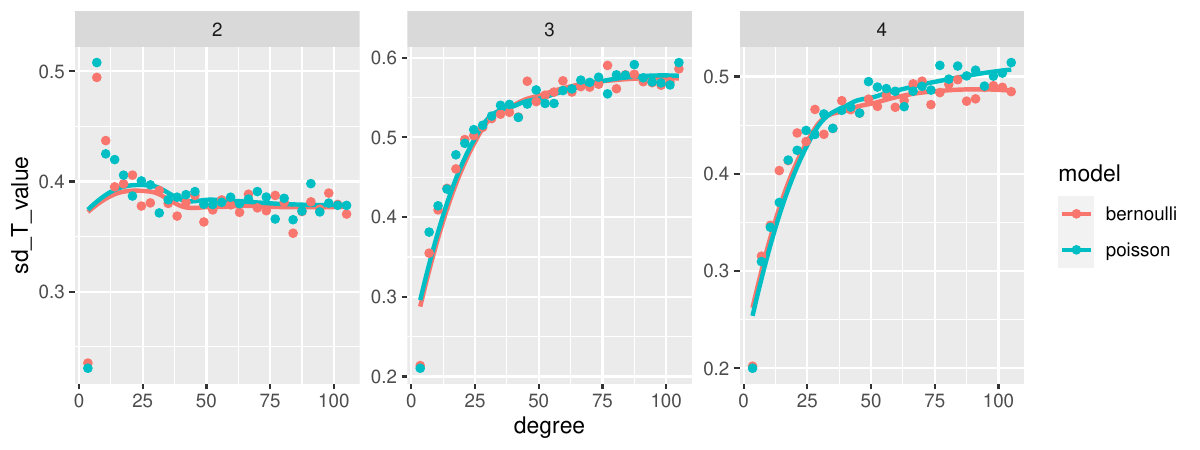} 
	\caption{ This is a repeat of Figure \ref{fig:sdsim}, but with $\textrm{folds} = 10$ instead of 1.  The vertical axis gives the standard deviation of the test statistics $T_2, T_3,T_4$ over the 1000 replicates.  Notice that they are all significantly less than one.   This explains why the rejection probabilities in Figure \ref{fig:rejection_prob_10fold} are significantly lower than .05.  
	}
	\label{fig:sdsim10fold} 
\end{figure}

\section{Additional information on the comparison of techniques}

In Figure \ref{fig:accuracy}, we evaluated the accuracy of each method when requiring the exact recovery of $k$. In order to illustrate how each method either under-estimates or over-estimates $k$, Figure \ref{fig:dev} displays the results in Figure \ref{fig:accuracy} by the relative error for each estimate $\hat k$, which is defined as
$$\text{relative error}=\frac{\hat{k}-k^*}{k^*},$$
where $k^*=10$ is the true $k$.
From the simulation results, we observed that most methods under-estimate $k$ when the average degree of the graph is smaller (i.e., sparser), except for StGoF which over-estimates it. 
In addition, from the standard deviation of the relative error, we observe that \texttt{EigCV} provides a more accurate and less variable estimation of $k$ as the graph sparsity varies.

\begin{figure}[ht]
	\centering
	\includegraphics[width=1\linewidth]{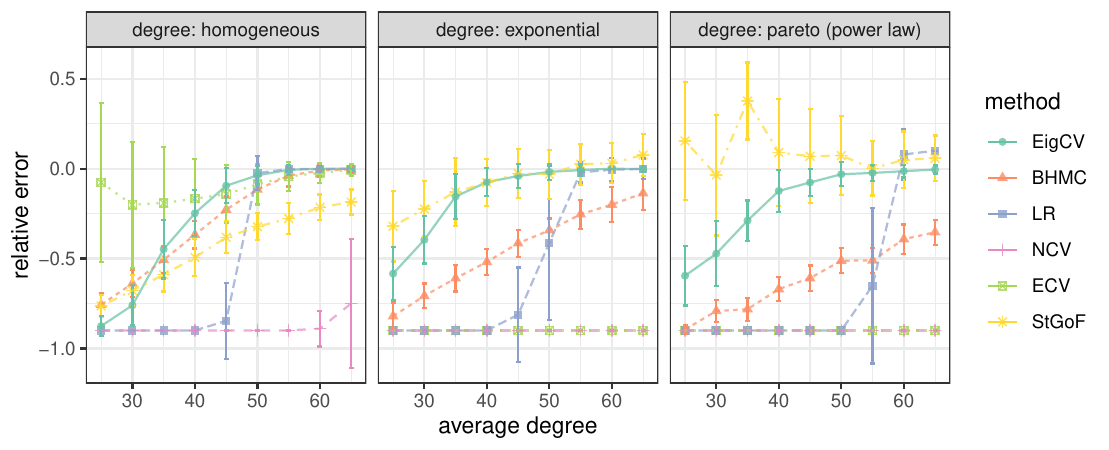}
	\caption{Comparison of relative error for different graph dimensionality estimators under the DCSBM. The panel strips on the top indicate the node degree distribution used. Within each panel, each colored line depicts the relative error of each estimation method as the average node degree increases. Each point on the lines are averaged across 100 repeated experiments. For each point, an error bar indicates the sample standard deviation of relative errors.}
	\label{fig:dev}
\end{figure}

In Section \ref{sec:email}, we removed the 14 small departments that consist of less than 10 members. Among these, two departments have only one members, and eight departments have less than five members. Table \ref{tbl:email_42} compares six methods using this email network without filtering. We observed similarly that \texttt{EigCV} provided a closer estimate of $k$ than other methods. 

\begin{table}[ht]
	\centering
	\caption{Comparison of graph dimensionality estimates using the email network among members in a large European research institution. Each members belongs to one of 42 departments.}
	\label{tbl:email_42}
	\smallskip
	\begin{tabular}{lrr}
		\hline
		Method & Estimate (mean) & Runtime (second) \\ 
		\hline 
		EigCV & 30.56 & 0.81 \\ 
		BHMC & 14.00 & 0.04 \\ 
		LR & 13.00 & 128.17 \\ 
		NCV & 6.96 & 271.15 \\ 
		ECV & 20.08 & 60.13 \\ 
		StGoF & $> 50$ & 544.66 \\ 
		\hline
	\end{tabular}
\end{table}

\clearpage

\footnotesize
\singlespace
\bibliographystyle{plainnat}
\bibliography{gdim}
	
\end{document}